\documentclass[a4paper]{article}
\pdfoutput=1

\usepackage{amsmath}
\usepackage{amsthm}
\usepackage{amssymb}
\usepackage{bm}
\usepackage{amsfonts}
\usepackage{mathrsfs}
\usepackage{mathtools}
\usepackage{color}
\usepackage[dvipsnames]{xcolor}

\usepackage[boxsize=0.3em,aligntableaux=center]{ytableau}

\newtheorem{thm}{Theorem}
\newtheorem{dfn}[thm]{Definition}
\newtheorem{prp}[thm]{Proposition}
\newtheorem{lem}[thm]{Lemma}

\theoremstyle{remark}
\newtheorem{rem}{Remark}

\usepackage{tikz-cd}
\usetikzlibrary{matrix,arrows,calc}

\tikzset{ampersand replacement=\&}

\usepackage[colorlinks=true,linkcolor=blue,citecolor=red,urlcolor=Violet,bookmarksdepth=subsection,final]{hyperref}

  \newcommand{\del}{\partial}
  \newcommand{\oo}{\infty}

  \newcommand{\id}{\mathrm{id}}
  
  \newcommand{\Secs}{\Gamma}
  \newcommand{\oast}{\circledast} 

  \newcommand{\ga}{\gamma}
  \newcommand{\La}{\Lambda}
  \newcommand{\la}{\lambda}
  
  \newcommand{\D}{\mathbb{D}}
  \newcommand{\KIDi}{\operatorname{KID}^{(1)}}
  \newcommand{\KIDii}{\operatorname{KID}^{(2)}}
  \newcommand{\uKIDo}{\operatorname{uKID}^{(0)}}
  \newcommand{\uKIDi}{\operatorname{uKID}^{(1)}}
  \newcommand{\uKIDii}{\operatorname{uKID}^{(2)}}

\renewcommand{\d}{\mathrm{d}}

\allowdisplaybreaks[1] 

\title{Unit Killing Initial Data}
\author{%
	Igor Khavkine$^\flat$%
		\thanks{E-mail: khavkine@math.cas.cz}
		{} and
	Salvador Mengual Sendra$^{\sharp}$%
		\thanks{E-mail: salvador.mengual@uv.es}
\\[2ex]
	{\small $^\flat$Institute of Mathematics of the Czech Academy of Sciences,}\\
	{\small \v{Z}itn{\'a} 25, 110 00 Praha 1, Czech Republic}\\
	{\small $^\sharp$Astronomy and Astrophysics Department, Universitat de Val{\`e}ncia, }\\
	{\small  Av. Vicent Andr{\'e}s Estell{\'e}s, 19, 46100, Burjassot, Spain}
}
\date{}

\begin{document}
\maketitle

\begin{abstract}
We find a system of differential equations on an initial data surface
whose solutions are in bijection with unit vector fields on the Einstein
$\Lambda$-vacuum development that are proportional to a Killing vector.
We refer to these conditions as the \textit{unit Killing initial data}
(uKID) equations, analogous to the classical \textit{Killing initial
data} (KID) equations. The uKID equations can be useful in a setting
where only the unit-normalized part of the Killing vector is
geometrically distinguished. We eliminate the scaling degree of freedom
of a general Killing vector to obtain the space-time equations
characterizing unit normalized Killing vector fields and also the uKID
equations. These equations are also prolonged to canonical connection
form, showing their finite type character. Finally, we obtain an
independent derivation of the uKID equations by revisiting the propagation
identity method, which has previously been used to characterize the
initial data of other geometric equations.
\end{abstract}

\section{Introduction}

The Killing equation characterizes the infinitesimal isometries of a
pseudo-Riemannian manifold, its Killing vector fields. A related notion
of \emph{Killing initial data (KID)} equations gained prominence in
mathematical relativity after being so named in~\cite{beig-chrusciel}.
They consist of a set of partial differential equations (PDEs) intrinsic
to a spacelike hypersurface $\Sigma \subset M$ in a Einstein vacuum
Lorentzian manifold $(M,g_{ab})$, whose special property is that their
solutions are in bijection (when $\Sigma$ is partial Cauchy) with the
Killing vector fields in the domain of dependence $D(\Sigma)$. In
particular, if $\Sigma$ is a Cauchy surface for $(M,g_{ab})$, so that $M =
D(\Sigma)$, solving the KID equations on $\Sigma$ allows one to check
whether the metric initial data will evolve into a spacetime with
Killing vectors, without first solving the Einstein equations. Moreover,
due to their geometric character, the KID equations can be formulated on
an arbitrary choice of Cauchy surface and coordinates thereon.

The existence of what is now known as the KID equations was first
noticed much earlier~\cite{berezdivin, MONCRIEF-KID1, MONCRIEF-KID2,
COLL77} and has since then inspired a number of
generalizations and applications. On the one hand, they have been extended to covariant gravity-matter systems, where KID-type equations characterize the existence of space-time symmetries from initial data and guarantee that these symmetries are preserved by the hyperbolic evolution of the coupled field equations~\cite{Racz-1999, Racz-2001}. On the other hand, they have become a fundamental tool in unique continuation problems and black-hole rigidity, allowing local Killing fields arising from characteristic initial data to be propagated into the space-time development~\cite{Ionescu-Klainerman}.

One application of the KID equations is to a geometric characterization
of initial data of the rotating Kerr black hole intrinsic to an
arbitrary Cauchy surface~\cite{gp-kerrid}. That characterization
combines a so-called \emph{IDEAL characterization} of the Kerr
geometry~\cite{IDEAL-kerr} with the KID equations and a way of
geometrically expressing Kerr's stationary Killing vector directly from
the curvature. It would be interesting to extend similar intrinsic
initial data characterizations to other physically interesting solutions
of Einstein equations with fluid matter~\cite{Volkoff, Wyman}. While these
solutions do possess static Killing vectors, they can only be
obtained from the curvature in unit normalized form. Already back in the
well-known reference~\cite[\textsection 2.2(vi)]{ehlers-mech}, Ehlers
noticed that a timelike unit vector that is a rescaling of a Killing
vector can be identified by a related set of equations and referred to
it as \emph{isometric}.%
	\footnote{Strictly speaking, Ehlers referred to the flow of such a
	vector field as an \emph{isometric flow}. We might as well refer to
	the vector field itself as isometric.} %
For reference, let us call these the \emph{unit
Killing equations}. A natural question is whether the initial data
characterization strategy from~\cite{gp-kerrid} can be generalized to
the situation where only a unit normalized Killing vector is available.
Part of a positive answer would be identifying the unit Killing initial
data (uKID) equations.

In this work, we precisely derive the uKID equations both independently
and following the \emph{propagation identity} framework outlined
in~\cite{khgp-ckid, khgp-cykid}. The former is presented in
Section~\ref{sec:killing}. In Section~\ref{sec:ukilling}, we start from
the Killing equation and we obtain the analogue system of equations
characterizing the unit part of a time-like and space-like Killing
vector fields. Then, in Section~\ref{sec:uKID} we introduce a space-time
split around a Cauchy surface and obtain the initial data conditions for
unit Killing vector fields. In Section~\ref{sec:prop}, we obtain the
propagation identities ensuring that the condition of unit Killing
vector field is maintained along the evolution. Besides providing an
another derivation of the uKID equations, these propagation identities
may be independently useful in stability analyses
like~\cite{Ionescu-Klainerman}. First, in Section~\ref{sec:propeq} we
develop the theoretical framework on which the result obtained in
Section~\ref{sec:ukilling-propeq} is based, with technical auxiliary
results on Agmon-Douglis-Nirenberg hyperbolicity relegated to
Appendix~\ref{app:adn}. We also obtain, for PDE systems that we derive
from the Killing equation, a prolonged connection form, justifying that
they inherit the property of being of \emph{finite type}. Finally, in
Section~\ref{sec:discussion} we discuss our results and outline
directions for future research.

\section{Unit Killing vectors and initial data} \label{sec:killing}

In this section, we obtain the unit Killing equations, both in their
spacetime and initial data versions. To make our conventions precise,
we make the blanket assumption that all manifolds, bundles and maps are
infinitely smooth. For our conventions on Riemann curvature we
follow~\cite{XACT}, which is the same as~\cite{wald-gr}.

\subsection{Unit Killing equations} \label{sec:ukilling}

Consider the Killing equation on a Lorentzian manifold $(M,g_{ab})$:
\begin{equation} \label{eq:killing}
	K_{ab}[v] := \nabla_a v_b + \nabla_b v_a = 0.
\end{equation}
Suppose now that $v_b$ is timelike and set $v_b := \la u_b$, where $u_b
u^b = -1$, hence $v_b v^b = \la^2$. Define the projections orthogonal
to $u_a$ as $\ga_{ab} = g_{ab} + u_a u_b$ and the covariant derivative
along $u_a$ as $\dot{(-)} = u^e \nabla_e (-)$.
Substituting in the Killing operator gives
\begin{equation} \label{eq:killing-lu}
	K_{ab}[\la u] = \la (K_{ab}[u] + 2 u_{(a} \nabla_{b)} \ln\la) ,
\end{equation}
with contractions
\begin{align}
	\la^{-1} u^a K_{ab}[\la u]
		&= u^a \nabla_a u_b + u^a \nabla_b u^a
			+ u_b u^a \nabla_a \ln\la - \nabla_b \ln\la
		\notag \\
		&= \dot{u}_b + \tfrac{1}{2} \nabla_b (u^a u_a)
			 - \nabla_b \ln\la + u_b \dot{\ln\la} ,
	\\
	\la^{-1} u^a u^b K_{ab}[\la u]
		&= \dot{u^a u_a} - 2 \dot{\ln\la} = -2 \dot{\ln\la} . 
\end{align}
Defining
\begin{align}
\label{eq:A}
	A_a[\la,u] &:= \dot{u}_a - \nabla_a \ln\la , \\
\label{eq:spK}
	\bar{K}_{ab}[u] &:= \gamma^c_a \gamma^d_b K_{cd}[u]
	= (\nabla_a + \dot{u}_a) u_b + (\nabla_b + \dot{u}_b) u_a , \\
\label{eq:B}
	B_{ab}[u] &:= (\d A[\la,u])_{ab} = 2\nabla_{[a} \dot{u}_{b]} ,
\end{align}
the Killing operator can now be rewritten
\begin{align} \label{eq:killing-Kb-A}
	K_{ab}[\la u]
	&= \la (\gamma^c_a \gamma^d_b K_{cd}[u]
		- 2 u_{(a} (\dot{u} - \nabla \ln\la)_{b)})
	\notag \\
	&= \la (\bar{K}_{ab}[u] - 2 u_{(a} A_{b)}[\la,u] ) .
\end{align}
This leads us to the following
\begin{prp}[Unit Killing equations] \label{prp:ukilling}
(a) Given a timelike unit vector field $u_a$ on $(M,g_{ab})$, the equations
\begin{equation} \label{eq:ukilling}
	\bar{K}_{ab}[u] = 0 \quad \text{and} \quad
	B_{ab}[u] = 0
\end{equation}
are satisfied iff there locally exists a scalar $\la$ such that $v_a
= \la u_a$ is a Killing vector, $K_{ab}[v] = 0$.

(b) When $u_a$ is a spacelike unit vector, the analogous equations are
\begin{equation} \label{eq:s-ukilling}
	\bar{K}^+_{ab}[u] := (\delta^c_a - u^c u_c) (\delta^d_b - u^d u_b) K_{cd}[u]
		= 0 \quad \text{and} \quad
	B_{ab}[u] = 0.
\end{equation}
\end{prp}

\begin{proof}
Equation $B_{ab}[u] = 0$ implies that locally $u_a$ is the gradient of a
scalar (Poincar\'e lemma), which we might as well denote $\ln\la$.
This implies $A_a[\la, u] = 0$. But the calculations
in~\eqref{eq:killing-lu}--\eqref{eq:killing-Kb-A} above imply that the
equations $\bar{K}_{ab}[u] = 0, A_a[\la, u] = 0$ and $K_{ab}[\la
u] = 0$ are equivalent.

When $u_a$ is spacelike, the argument is completely analogous, with the
slightly modified key identity $K_{ab}[\la u] = \la (\bar{K}_{ab}[u] + 2
u_{(a} A^+_{b)}[\la,u])$, where $A^+_a[\la,u] := \dot{u}_a +
\nabla_a \ln\la$. As a shortcut, the appropriate equations can be
obtained by noting that $B_{ab}[iu] = -B_{ab}[u]$ and $\bar{K}_{ab}[iu]
= i\bar{K}^+_{ab}[u]$, where $i^2=-1$.
\end{proof}

We will consistently refer to~\eqref{eq:ukilling}
or~\eqref{eq:s-ukilling} as the \emph{unit Killing equations}, though
following~\cite[\textsection 2.2(vi)]{ehlers-mech} they could be equally
well called the \emph{isometric vector equations}. Notably, unlike the
Killing equations, the unit Killing equations are non-linear, which is
due to the non-linearity of the parametrization $v_a = \lambda u_a$ with
unit $u_a$.

It is worth remarking on the \emph{type} of PDE that the unit Killing
system belong to. It is well-known that the Killing equation is an
example of a PDE system of \emph{finite type}, which is characterized by
the property that in any sufficiently small neighborhood of a point, all
solutions are uniquely determined by data of finite differential order
at that point, (cf.~\cite[Def.7, App.A]{kh-compat} and references
therein for a discussion of this definition). It is this property that
is responsible for the well-known fact that an $n$-dimensional manifold
has at most $n(n+1)/2$ independent Killing vectors. The clearest way to
establish the finite type property of a PDE system is to show that it is
equivalent to the equation of parallel sections of some bundle with
respect to a connection (the prototypical equation from this class),
with the rank of the bundle determining the maximal dimension of the
space of solutions at a point. We refer to such an equivalence as a
\emph{prolongation to canonical form}, and to be precise by
\emph{equivalent} we mean that there exist differential operators in
both directions that establish a bijection on the solutions of the two
PDE systems. For the Killing equation, that canonical form is known as
the \emph{Kostant connection} form~\cite{kostant} or the \emph{Killing
transport equations}~\cite{geroch} (cf.~\cite[Sec.3]{kh-compat} for a
clean derivation):
\begin{equation} \label{eq:cankilling}
	\nabla_a v_b = \omega_{ab}, \quad
	\nabla_a \omega_{bc} = R^d{}_{abc} v_d ,
\end{equation}
where antisymmetrizing the first equation gives $\omega_{ab} =
\nabla_{[a} v_{b]}$, and plugging that in recovers the original equation
on $v_a$.

Being a kind of projection of the Killing equation, the unit Killing
system should also be of finite type. This can be independently verified
from the following prolonged canonical form:
\begin{thm}[prolongued unit Killing equations] \label{thm:canukilling}
The following system is equivalent to \eqref{eq:ukilling}:
\begin{equation} \label{eq:canukilling}
\begin{aligned}
	\nabla_a u_b &= \omega_{ab} + 2 u^c u_{(a}\omega_{b)c}, \\
	\nabla_c \omega_{ab} &=
		u^d (2 \omega_{d[a}\omega_{b]c} + 2 \omega_{ab}\omega_{cd} + 3 u^e u_{[b}\omega_{a]d}\omega_{ce} + u_c u_{[a}\omega_{b]}{}^e \omega_{de})
		\\ & \quad {}
		+ u_{[b}\omega_{a]}{}^d \omega_{cd}
		+ u^d (u_{[a} R_{b]dce} u^e - R_{abcd})
	.
\end{aligned}
\end{equation}
The original equations on $u_a$ are recovered by antisymmetrizing the first
equation to get
\begin{equation} \label{eq:canukilling-defs}
	\omega_{ab} = \nabla_{[a}u_{b]},
\end{equation}
and plugging it in.
\end{thm}

\begin{rem}
From the rank of the auxiliary bundle in the prolonged canonical
forms~\eqref{eq:cankilling} and~\eqref{eq:canukilling}, we can tell that
the maximum dimension of the solution space: the well-known $n(n+1)/2$
for Killing and $n(n+1)/2-1$ for unit Killing (as $u_a$ is unit
normalized).
\end{rem}

\begin{proof}
One direction is straightforward since it can be checked that a solution of~\eqref{eq:canukilling} solves the unit Killing equations by mechanically substituting~\eqref{eq:canukilling} into~\eqref{eq:ukilling} and eliminating $\omega_{ab}$ using the definition~\eqref{eq:canukilling-defs}. In the other direction, given a solution $u_a$ of
the unit Killing system, it solves the prolonged system~\eqref{eq:canukilling} when
the extra $\omega_{ab}$ component is defined
by~\eqref{eq:canukilling-defs}. Up to using this definition, the $\nabla_a u_b$
equation coincides with $\frac12 \bar{K}_{ab}[u]$, while the $\nabla_c \omega_{ab}$ equation coincides with
$F_{abc}[u] + 2 u_c G_{[ba]}[u] + 2G_{c[a}[u] u_{b]}$, where $F_{abc}[u] := \nabla_{[a} \bar{K}_{b]c}[u] + u^d \omega_{d[a} \bar{K}_{b]c}[u] + (u_c \omega_{d[a} - \omega_{cd} u_{[a}) \bar{K}_{b]}{}^d [u] $, $G_{ab}[u] := \frac12 u^c F_{acb}[u] + \frac12 J_{ab}[u] - u^c u_{(a} J_{b)c}[u]$ and $J_{ab}[u] := \frac12 B_{ab}[u] + \frac12 ( \frac12\bar{K}_{[a}{}^c [u] - \frac12 \nabla_{[a} u^c + \nabla^c u_{[a}) \bar{K}_{b]c}[u]$. Hence, $u_a \mapsto (u_a, \nabla_{[a} u_{b]})$ and $(u_a, \omega_{ab}) \mapsto u_a$ clearly give a bijection on solutions.
\end{proof}

\subsection{Unit Killing initial data} \label{sec:uKID}

Let $\Sigma \subset M$ be a spacelike hypersurface. We will use capital
Latin letters $A, B, C, \dots$ for tensors intrinsic to $\Sigma$. With
the metric induced from $(M,g_{ab})$ it is a Riemannian manifold
$(\Sigma, g_{AB})$ equipped with symmetric tensor $\pi_{AB}$, the
\emph{second fundamental form}.
The triple $(\Sigma, g_{AB}, \pi_{AB})$ can be interpreted as the
initial data for the Einstein equations.
To fix conventions, denote by $D_A$ be the Levi-Civita connection on
$(\Sigma,g_{AB})$, while the spatial Riemann and Ricci tensors are
denoted $r_{ABCD}$ and $r_{AB}$. Let $t$ be the future increasing time
function of \emph{Gaussian normal coordinates} defined locally around
$\Sigma$. That is, $\Sigma = \{ t = 0 \}$, while $n_a = \nabla_a t$ is
timelike unit normal, geodesic and orthogonal to $\Sigma$. The extrinsic
curvature is then defined as  
\begin{equation}
	\pi_{AB} = \frac12 \mathcal{L}_n g_{AB} ,
\end{equation}
where $\mathcal{L}_n$ denotes the Lie derivative with respect to $n$.
Any vector can be decomposed as $v_a = v_0
(\nabla_a t) + h_a^A v_A$, where $h_a^A$ is the inclusion map of spatial
vectors into spacetime ones. Analogous decompositions hold for higher
rank tensors.

Now, when the spacetime satisfies the $\La$-vacuum Einstein equations
$R_{ab} = \frac{2\La}{n-2} g_{ab}$, which can be decomposed into 
\begin{equation} \label{eq:La-vac}
\begin{aligned}
	\nabla_0 \pi_{AB} + \pi \pi_{AB} + r_{AB} &= \frac{2\La}{n-2} g_{AB} , \\
	D_C \pi^C{}_A - D_A \pi &= 0 , \\
	\nabla_0 \pi_C{}^C + \pi_{CD} \pi^{CD} &= \frac{2\La}{n-2} ,
\end{aligned}
\end{equation}
where $\nabla_0 = n^a \nabla_a$, the \emph{Killing
Initial Data (KID) equations} are~\cite[Eq.(20)]{khgp-cykid}
\begin{equation} \label{eq:KID}
\begin{aligned}
	\KIDi_{AB}[v] &:= D_A v_B + D_B v_A + 2\pi_{AB} v_0 = 0 ,
	\\
	\KIDii_{AB}[v] &:= D_{A} D_{B} v_{0}
	+ (2(\pi\cdot\pi)_{A B} + \nabla_0 \pi_{AB}) v_0 \\
	& \quad {}
	+ 2 \pi_{(B}{}^{C} D_{A)} v_C
	+ (D^{C}{\pi_{A B}}) v_{C} = 0 ,
\end{aligned}
\end{equation}
where $(v_0, v_A)$ are the components of a spacetime vector $v_a$ on
$\Sigma$ that will (at least locally) give rise to a Killing vector of
$(M,g_{ab})$ under a certain extension away from $\Sigma$.

\begin{rem}
Assuming that the Einstein equations hold on a neighborhood of $\Sigma
\subset M$, we could express any time derivatives of $\pi_{AB}$ in terms
of $g_{AB}$ and $\pi_{AB}$. However, that would rely on picking a
specific matter model and stress energy tensor (e.g.,\ vacuum, perfect
fluid, etc.). Guided by our motivation to consider different matter
models in the future, we have decided not to expand $\nabla_0 \pi_{AB}$
in that way. Our version of the KID equations~\eqref{eq:KID} hence
actually depend on the \emph{extended} initial data quadruple $(\Sigma,
g_{AB}, \pi_{AB}, \nabla_0\pi_{AB})$, where the last element is a new
independent symmetric tensor intrinsic to $\Sigma$, with $\nabla_0$ only
appearing formally as a suggestive notation. These KID equations should
then be valid provided one of the methods of establishing an isomorphism
between their solutions on $\Sigma$ and the Killing vectors on a
neighborhood of $\Sigma$ holds, either in vacuum~\cite{MONCRIEF-KID1,
COLL77, beig-chrusciel, khgp-ckid} or with certain kinds of
matter~\cite{racz-kid1, racz-kid2}.
\end{rem}

The rescaling $v_a \mapsto \la v_a$ is equivalent to $(v_0, v_A)
\mapsto (\la v_0, \la v_A)$ and the normalization $u_a u^a = -1$ is
equivalent to $u_A u^A = u_0^2 - 1$. Choosing $\la$ so that $v_a =
\la u_a$, with $u_a$ unit normalized, gives%
\begin{equation} \label{eq:KID-lu}
\begin{aligned}
	\la^{-1} \KIDi_{AB}[\la u]
	&= \KIDi_{AB}[u] + u_A D_B \ln\la + u_B D_A \ln\la ,
	\\
	\la^{-1} \KIDii_{AB}[\la u]
	&= \KIDii_{AB}[u] + u_0 \la^{-1} D_A D_B \la
		+ 2 (D_{(A} u_0) (D_{B)} \ln\la)
	\\ & \quad {}
		- 2 u_C \pi^C{}_{(B} D_{A)} \ln\la
	,
	\\
	&= \KIDii_{AB}[u] + u_0 [D_A D_B \ln\la + (D_A\ln\la) (D_B\ln\la)]
	\\ & \quad {}
		+ 2 (D_{(A} u_0) (D_{B)} \ln\la)
		- 2 u_C \pi^C{}_{(B} D_{A)} \ln\la
	,
\end{aligned}
\end{equation}
and the contractions
\begin{equation}
\begin{aligned}
	u^A u^B \la^{-1} \KIDi_{AB}[\la u]
	&= u^A u^B \KIDi_{AB}[u] + 2 (u_C u^C) u^B D_B \ln\la
	,
	\\
	u^B \la^{-1} \KIDi_{AB}[\la u]
	&= u^B \KIDi_{AB}[u] + u_A (u^B D_B \ln\la)
	\\ & \quad {}
		+ (u_C u^C) D_A \ln\la
	.
\end{aligned}
\end{equation}
These equations lead us to define the operator
\begin{equation} \label{eq:uKID-L}
	L_{A}[v] := 
	\frac{v^B \KIDi_{AB}[v]}{(v_C v^C)}
		- \frac{v_A (v^B v^C \KIDi_{BC}[v])}{2(v_C v^C)^2}
	,
\end{equation}
which obviously factors through the KID equations and has the key property
\begin{equation}
	L_A[\la u] = L_A[u] + D_A \ln\la .
\end{equation}
Of course, $L_a[v]$ is defined only where $v_A$ is non-zero.
Next, let us define the following operators acting on a unit vector
$(u_0,u_A)$:
\begin{equation} \label{eq:uKID-def}
\begin{aligned}
	\uKIDo_{AB}[u] &:= 2D_{[A} L_{B]}[\la u] , \\
	\uKIDi_{AB}[u] &:= \la^{-1} \KIDi_{AB}[\la u] - 2 u_{(A} L_{B)}[\la u] , \\
	\uKIDii_{AB}[u] &:= \la^{-1} \KIDii_{AB}[\la u]
		\\ &\quad {}
		- u_0 \left(D_{(A} L_{B)}[\la u] + L_A[\la u] L_B[\la u] - 2 L_{(A}[u] L_{B)}[\la u]\right)
		\\ &\quad {}
		- 2L_{(A}[\la u] D_{B)} u_0
		+ 2u_C \pi^C{}_{(A} L_{B)}[\la u] ,
\end{aligned}
\end{equation}
where the $\la$-dependence on the right-hand side completely cancels
(otherwise the definition would not make sense). A $\la$-free definition
is obtained by simply setting $\la = 1$ with unit $(u_0,u_A)$. Note that
one term in $\uKIDii_{AB}[u]$ contains both $L_A[u]$ and $L_A[\lambda
u]$, the important point is that it depends at least linearly on
$L_A[\lambda u]$. After setting $\lambda = 1$, the distinction between
these expressions disappears.

Analogous to Proposition~\ref{prp:ukilling}, we obtain
\begin{thm}[Unit Killing initial data] \label{thm:uKID}
Let $(\Sigma, g_{AB}, \pi_{AB}, \nabla_0\pi_{AB})$ be an extended
initial Einstein data surface.
(a) Given a non-vanishing spatial vector field $u_A$ on $\Sigma$, setting
$u_0^2 = 1 + u_A u^A$, the equations
\begin{equation} \label{eq:uKID}
	\uKIDo_{AB}[u] = 0 , \quad
	\uKIDi_{AB}[u] = 0 , \quad
	\uKIDii_{AB}[u] = 0
\end{equation}
are satisfied iff there locally exists a scalar $\la$ such that $(v_0,
v_A) = (\la u_0, \la u_A)$ is timelike, non-orthogonal to $\Sigma$ and
satisfies the KID equations~\eqref{eq:KID}.

(b) If $u_A u^A \ge 1$ and we set $u_0^2 = -1 + u_A u^A$ in the defining
equations~\eqref{eq:uKID-def}, then the equations~\eqref{eq:uKID}
are satisfied iff there locally exists a scalar $\la$ such that $(v_0,
v_A) = (\la u_0, \la u_A)$ is spacelike and satisfies the KID
equations~\eqref{eq:KID}.
\end{thm}
\begin{proof}
The proof is completely analogous to that of
Proposition~\ref{prp:ukilling} and follows from the calculations
in~\eqref{eq:KID-lu}--\eqref{eq:uKID-def}, with $L_A[\la u]$ playing the
role of $A_a[\la,u]$. Note also that the substitution $u_A \mapsto i
u_A$ in~\eqref{eq:uKID} yields the same equations up to a factor of $i$,
but $u_0$ given by the formula in part (b).
\end{proof}
We will refer to~\eqref{eq:uKID} as the \emph{unit Killing initial data
(uKID) equations}. Again, like in the remark following
Proposition~\ref{prp:ukilling}, unlike the KID equations, the uKID
equations are non-linear.

Not surprisingly, the KID and uKID systems are also PDEs of finite type
(see the discussion in Section~\ref{sec:ukilling}). For the record, we
give their prolongations to canonical form:

\begin{thm}[prolongued KID and uKID equations] \label{thm:canuKID}
(a) The following system is equivalent to \eqref{eq:KID}:
\begin{equation} \label{eq:canKID}
\begin{aligned}
D_A v_B &= -\pi_{AB} v_0 + \omega_{AB} , \\
D_A v_0 &= p_A , \\
D_A \omega_{BC} &= 2 (\pi_{A[B} p_{C]} - v_0 D_{[B} \pi_{C]A}) - r_{BCAD}
u^D , \\
D_A p_B &= -2\pi_{(A}{}^C \omega_{B)C} - u^C D_C \pi_{AB} - v_0 \nabla_0 \pi_{AB} .
\end{aligned}
\end{equation}
The original equations on $v_A$ are recovered by using the following
substitutions, obtained from the first two equations above:
\begin{equation} \label{eq:canKID-defs}
	\omega_{AB} = D_{[A} v_{B]} , \quad
	p_A = D_A v_0 .
\end{equation}

(b) The following system is equivalent to \eqref{eq:uKID}:
\begin{equation} \label{eq:canuKID}
\begin{aligned}
	D_A u_B &=
		\omega_{AB} - \pi_{AB}\,u_0 + 2 q_{(A} u_{B)},
\\
	D_A q_B &=
		\frac12 \omega_A{}^{C}\omega_{BC} + \frac52 q_A q_B + \omega_{C(A} u_{B)} q^{C} + \frac12 q_C q^{C} u_A u_B 
\\
& \quad
		+ \frac12 (\pi_{AC} \pi_{BD} - \pi_{AB} \pi_{CD} - r_{ACBD} ) u^{C} u^{D} - \frac{\pi_{AB} q_C u^{C}}{u_0}
\\
& \quad
		- \frac{1}{2 u_0^2} (q_A q_B + q^{C} q^{D} u_A u_B u_C u_D + 2\omega_{D(A} q^{C} u_{B)} u_C u^{D}
\\
& \quad
		- 2q^{C} u_C q_{(A} u_{B)} - 2\omega_{C(A} q_{B)} u^{C} + \omega_{AC} \omega_{BD} u^{C} u^{D} )		
\\
& \quad
		- u_0u^{C} ( D_{(A} \pi_{B)C} + D_C \pi_{AB} )
\\
& \quad
		+ \frac12 u_0^2 (\nabla_0 \pi_{AB} + \pi_A{}^{C} \pi_{BC}),
\\
	D_A \omega_{BC} & = 
		2 \omega_{BC} q_A + 2 \omega_{A[C} q_{B]} + \omega_A{}^D \omega_{D[C} u_{B]} - 2 q_A q_{[C} u_{B]}
\\
& \quad
		 + \pi_{A[C} \pi_{B]D} u^D - r_{BCAD} u^D + r_{ADE[B} u_{C]} u^D u^E
\\
& \quad		
		+ \pi_{DE} \pi_{A[C} u_{B]} u^D u^E + q^D \omega_{D[C} u_{B]} u_A
\\
& \quad
		+ \pi_{AD} \pi_{E[B} u_{C]} u^D u^E - 2 \pi_{A[C} q_{B]}
\\
& \quad
		- \frac{1}{u_0} (2 \pi_{A[C} \omega_{B]D} u^D - 2 \pi_{A[C} q_{B]} + q_A q_{[C} u_{B]})
\\
& \quad
		+ \frac{1}{u_0^2}(\omega_{D[C} u_{B]} q_A u^D - \omega_{AD} q_{[C} u_{B]} u^D - q^D q_{[C} u_{B]} u_A u_D 
\\
& \quad	
		+ q_A q_{[C} u_{B]} + q^D\omega_{E[B} u_{C]} u_A u_D u^E - \omega_{AD} \omega_{E[C} u_{B]} u^D u^E)  
\\
& \quad
		- u_0 [ (D_A \pi_{D[B}) u_{C]} u^D + 2D_{[B} \pi_{C]A} + u^D u_{[C} (D_{B]} \pi_{AD}) 
\\
& \quad
		+ 2 (D_D \pi_{A[C}) u_{B]} u^D + 2 \pi_{A[C} q_{B]}]
\\
& \quad
		- u_0^2 (\nabla_0 \pi_{A[C} + \pi_A{}^D \pi_{D[C}) u_{B]}.
\end{aligned}
\end{equation}
The original equations on $u_A$ are recovered using the following
substitutions, obtained from the first equation above:
\begin{equation} \label{eq:canuKID-defs}
\begin{aligned}
	\omega_{AB} &:= D_{[A}u_{B]}, \\
	q_A &:= \left( \frac{\pi_{AB} u^B}{u^D u_D} - \frac{\pi_{BC} u^B u^C}{2(u^D u_D)^2} u_A \right) u_0 + \frac{u^B(D_{(A} u_{B)})}{u^D u_D} - \frac{u^B u^C (D_C u_B)}{2(u^D u_D)^2} u_A .
\end{aligned}
\end{equation}
\end{thm}

\begin{proof}
(a) One direction is straightforward. Given a solution
of~\eqref{eq:canKID}, its $v_0$ and $v_B$ components solve the KID
system~\eqref{eq:KID}, which can be checked by mechanically
substituting~\eqref{eq:canKID} into~\eqref{eq:KID} and eliminating
$\omega_{AB}$ and $p_A$ using the definitions~\eqref{eq:canKID-defs}. A
quick inspection of definitions~\eqref{eq:canKID-defs} shows that they
themselves follow from~\eqref{eq:canKID}, so they are not separate
equations. In the other direction, given a solution $v_0$ and $v_B$ of
the KID system, it solves the prolonged system~\eqref{eq:canKID} when
the extra $\omega_{AB}$ and $p_A$ components are defined
by~\eqref{eq:canKID-defs}. Up to using these definitions, the $D_A v_B$
equation coincides with $\frac{1}{2} \KIDi_{AB}[v]$, the $D_A v_0$
equation is trivial, the $D_A \omega_{BC}$ equation coincides with
$D_{[B} \KIDi_{C]A}[v]$, while the $D_A p_B$ equation coincides with
$\KIDii_{AB}[v] - \pi_{(A}{}^C \KIDi_{B)C}[v]$. These maps clearly give
a bijection on solutions.

(b) The proof is essentially analogous to part (a), though
computationally more involved. The straightforward direction and the
bijectivity follow the same logic. In the other direction, we get a
solution of the prolonged system~\eqref{eq:canuKID} from a solution of
the uKID system~\eqref{eq:uKID} when the extra $\omega_{AB}$ and $q_A$
components are defined by~\eqref{eq:canuKID-defs}. Up to using these
definitions
\begin{align*}
	(D_A u_B + \cdots)
		&= \uKIDi_{AB}[u] , \\
	(D_{[A} q_{B]} + \cdots)
		&= \frac{1}{4} \uKIDo_{AB}[u] + O(\uKIDi[u]) , \\
	\MoveEqLeft
	(D_A \omega_{BC} - 2 u_{A} D_{[B} q_{C]} + 2 u_{[B} D_{C]} q_{A} + \cdots)
		\\
		&= 4 D_{[C} \omega_{B]A}  - 3 D_{[C} \omega_{BA]} + O(\uKIDi[u]) , \\
	\MoveEqLeft
	(D_{(A} q_{B)} + \frac{1}{2} (u_{A} u^C D_{[B} q_{C]} + u_{B} u^C D_{[A} q_{C]}) + \frac{1}{2} u^C D_{(A} \omega_{B)C} + \cdots)
		\\
		&= \frac{u_0}{2} \uKIDii_{AB}[u] + O(\uKIDi[u]) ,
\end{align*}
where $\cdots$ stand for terms of order zero and the omitted terms on
the right-hand side that are not explicitly shown are long and otherwise
unenlightening. From the structure of the left-hand sides above, it is
clear how to solve for $D_A u_B$, $D_A q_B$ and $D_A \omega_{BC}$ by
algebraic operations.
\end{proof}

\begin{rem}
From the ranks of the prolonged equations in~\ref{thm:canuKID} it
immediately follows that the maximum dimensions of the solution spaces
are $n(n+1)/2$ for the KID system and $n(n+1)/2-1$ for the uKID system,
when $\dim \Sigma = n-1$.
\end{rem}

\subsubsection{The orthogonal case}

The discussion leading up to Theorem~\ref{thm:uKID} assumed throughout
that $u_A u^A = u_0^2 - 1 \ne 0$. If that holds, then we are essentially
setting $(v_0, v_A) = (\la, 0)$ in~\eqref{eq:KID}:
\begin{equation}
\begin{aligned}
	2\pi_{AB} \la &= 0 ,
	\\
	D_{A} D_{B} \la
	+ (2(\pi\cdot\pi)_{A B} + \nabla_0 \pi_{AB}) \la  &= 0 ,
\end{aligned}
\end{equation}
and we would like to find the conditions on $g_{AB}$ and $\pi_{AB}$ such
that these equations are solvable for some $\la \ne 0$. The first equation
immediately implies $\pi_{AB} = 0$, while the second equation simplifies
to
\[
	D_{A} D_{B} \la + \la \nabla_0 \pi_{AB} = 0 .
\]
Upon substituting $\nabla_0 \pi$ by its $\La$-vacuum
form~\eqref{eq:La-vac}, we obtain an equation that has been studied in the
classification of vacuum static spacetimes, where $\la$ plays the role of
the lapse. We will not study this equation further here;
see~\cite{lafontaine} and the literature cited therein for more
information.

\section{Propagation identity} \label{sec:prop}

In this section, following~\cite[Sec.2]{khgp-ckid} we recall what we
mean by a \emph{propagation identity} for a differential equation
$E[\phi]=0$, with the Killing equation and its unit variant being
examples of our primary concern. However, for our current purposes, we
must generalize the key lemma~\cite[Lem.2.1]{khgp-ckid} that allows us
to connect such a propagation identity with $E$-initial data conditions
(Section~\ref{sec:propeq} below). Then in
Section~\ref{sec:ukilling-propeq}, we derive a covariant propagation
identity for the unit Killing equations~\eqref{eq:ukilling}. Its
evolution operators turn out to be of a ADN strictly hyperbolic type, an exotic
type of hyperbolicity that we recall in Appendix~\ref{app:adn}. The level of
generality introduced in Section~\ref{sec:propeq} was needed to cover the
connection of this kind of propagation identity to initial data equations.

\subsection{Propagation argument} \label{sec:propeq}

Consider a partial differential equation $P[\psi] = 0$, where we do not
require that it is linear or is of any specific differential order.
Under such general assumptions, some field configurations $\psi$ might
be pathological, so we restrict some properties to hold only for $\psi
\in \mathcal{A}$, a set of \emph{$P$-admissible} field configurations.
We say that it has a \emph{well-posed initial value problem} with
respect to a time function $t$ if there exists an \emph{initial data}
differential operator $ID[\psi]$ such that the values of $\Psi_0 =
ID[\psi]|_\Sigma$, with $\Sigma = \{ x\in M \mid t(x) = 0 \}$ and $\psi
\in \mathcal{A}$, uniquely determine an admissible solution $\psi$ on a
neighborhood $U \supset \Sigma$. We can interpret $ID[\psi]$ in two
ways: (i) as a differential operator on $M$, or (ii) as a differential
operator $ID[(\psi)|_\Sigma]$ on $\Sigma$ acting on $(\psi)|_\Sigma =
(\psi|_\Sigma, \del_t \psi|_\Sigma, \del_t^2 \psi|_\Sigma, \ldots)$,
which includes formal time derivatives of $\psi$ up to the necessary
order. Moreover, there exists differential operators $C$ and $IC$
deduced from $P[\psi] = 0$ itself, such that $IC[ID[(\psi)|_\Sigma]] =
C[P[\psi]]|_\Sigma$ for any $\psi$, with $C[0] = 0$ (so that, in
particular, the initial data of any actual solution automatically
satisfy $IC[\Psi_0] = 0$), and, for $\Psi_0$ induced by an admissible
configuration, the \emph{initial constraint conditions} $IC[\Psi_0] = 0$
ensure that there exists an admissible solution $\psi$ satisfying
$ID[\psi]|_\Sigma = \Psi_0$ on some neighborhood $U \supset \Sigma$.
Really, $IC[\Psi_0]$ is just a rewriting of the $C[P[\psi]]$ operator,
in terms of the initial data $\Psi_0 = ID[\psi]$, which of course
presumes that $ID[\psi]$ was chosen to expose enough data of $\psi$ to
make that possible without further differentiation in time. Finally, we
call $P[\psi] = 0$ a \emph{propagation equation} when it has a
well-posed initial value problem, as above, the vanishing configuration
is admissible, $0 \in \mathcal{A}$, and also the data $\Psi_0 = 0$
satisfy $IC[0] = 0$ and give rise to the solution $\psi = 0$. Just as
$P$ itself, the operators $C$, $IC$ and $ID$ may be non-linear. It is
possible that the operator $P[\phi;\psi]$ depends on some parameter
fields $\phi$. Because of the dependence on parameters, the
$P$-admissible field configurations consist of pairs $(\phi, \psi) \in
\mathcal{A}$, with $(\phi,0)$ admissible whenever some $(\phi,\psi)$ is
admissible.

The wave equation $P[\psi] = \square \psi = 0$ on a globally hyperbolic
spacetime with a Cauchy time function $t$ provides a standard example of
what we have called a propagation equation, where $ID[\psi] = (\psi,
\del_t \psi)$, while $C[\eta] = 0$ and $IC[(\psi)|_\Sigma] = 0$. The
complexity of our definition pays off only for more exotic examples,
such as the ADN strictly hyperbolic systems that we discuss in
Appendix~\ref{app:adn}.

The following generalizes Lemma~2.1 of~\cite{khgp-ckid}:
\begin{lem} \label{lem:propeq}
Let $E[\phi] = \psi$ be a PDE (system) on a manifold $M$, defined on
some (possibly multicomponent) fields $\phi$ and $\psi$. Let $t$ be a
time function defining a hypersurface $\Sigma = \{ x\in M \mid t(x) =
0 \}$, with respect to which a PDE $Q[\phi] = 0$ has a well-posed
initial value problem, with admissible set $\mathcal{A}_Q$, and a PDE
$P[\phi;\psi] = 0$ is a propagation equation on $\psi$, with admissible
set $\mathcal{A}_P$ of pairs $(\phi,\psi)$, such that $(\phi,E[\phi]) \in \mathcal{A}_P$ if $\phi \in \mathcal{A}_Q$. Let us call a configuration
$\phi$ \emph{admissible} when $\phi \in \mathcal{A}_Q$ and $(\phi,
E[\phi]) \in \mathcal{A}_P$. Suppose that $P$ and $Q$ also satisfy the
\emph{propagation identities}
\begin{equation} \label{eq:propeq}
	P[\phi; E[\phi]] = \sigma[\phi; Q[\phi]]
	\quad \text{and} \quad
	Q[\phi] = \rho[\phi; E[\phi]] ,
\end{equation}
with $P[\phi;0] = 0$, for some differential operators
$\sigma[\phi;\chi]$ and $\rho[\phi;\psi]$ also satisfying
$\sigma[\phi;0] = 0$ and $\rho[\phi;0] = 0$. Then, the unique admissible
solution $\phi$ of $Q[\phi] = 0$ (possibly defined only on a
neighborhood $U \supset \Sigma$) determined by a valid
initial data $ID_Q[(\phi)|_\Sigma] = \Phi_0$ ($IC[\Phi_0] = 0$)
satisfies $E[\phi] = 0$ on $U$ iff the initial data
$ID_P[(E[\phi])|_\Sigma] = 0$ vanish. Moreover, every admissible
solution of $E[\phi] = 0$ on a neighborhood $U \supset \Sigma$ arises in
this way.

In addition, there exists a purely spatial PDE on $\Sigma$,
$E^\Sigma[(\phi)|_\Sigma] = 0$ such that, whenever $Q[\phi] = 0$, the
condition $E^\Sigma[(\phi)|_\Sigma] = 0$ is equivalent to the vanishing
of the initial data $ID_P[(E[\phi])|_\Sigma] = 0$.
\end{lem}

\begin{proof}
We closely follow the logic of the proofs of Lemmas 2.1--2
of~\cite{khgp-ckid}, with adjustments for the more general hypotheses.

Let $\phi$ be the admissible solution of $Q[\phi] = 0$ on $U \supset
\Sigma$ determined by the valid initial data $\Phi_0$, and let $\psi =
E[\phi]$, so that the pair $(\phi,\psi)$ is $P$-admissible and, since
$P$ is a propagation equation, so is $(\phi,0)$. The first propagation
identity in~\eqref{eq:propeq} then
implies $P[\phi; \psi] = \sigma[\phi; Q[\phi]] = \sigma[\phi; 0] = 0$,
so that $\psi$ solves the propagation equation $P[\phi; \psi] = 0$.
Hence, using that $P[\phi;0] = 0$, if $\Psi_0 = ID_P[(\psi)|_\Sigma] =
0$, these initial data are automatically valid ($IC_P[\Psi_0] = IC_P[0]
= 0$) and, both $\psi$ and $0$ being admissible solutions with the same
initial data, the uniqueness part of $P$-well-posedness gives $E[\phi]
= \psi = 0$ on $U$. Conversely, of course, when $\psi = E[\phi] = 0$,
again by the uniqueness part of $P$-well-posedness it must have
$ID_P[\psi] = 0$.

For the second part, note that each formal time derivative $\del_t^N
E[\phi]|_\Sigma$ is itself a purely spatial differential operator
applied to $(\phi)|_\Sigma$. Hence, interpreting $ID_P$ as a
differential operator on $\Sigma$, the composition
$E^\Sigma[(\phi)|_\Sigma] := ID_P[(E[\phi])|_\Sigma]$ defines a purely
spatial PDE on $\Sigma$, tautologically equivalent to the vanishing of
the initial data $ID_P[(E[\phi])|_\Sigma] = 0$. Of course, whenever
$Q[\phi] = 0$, the equation itself and its formal time derivatives
restricted to $\Sigma$ may be used to eliminate some of the components
of $(\phi)|_\Sigma$ from $E^\Sigma$, replacing the tautological choice
above by a simplified but, modulo $Q[\phi] = 0$, equivalent one.

Finally, the remaining implication needed to prove the claim of the first
paragraph of the lemma follows from the second propagation identity
in~\eqref{eq:propeq}. Namely, any admissible solution of $E[\phi]
= 0$ satisfies $Q[\phi] = \rho[\phi; E[\phi]] = \rho[\phi; 0] = 0$,
with its own initial data $\Phi_0 = ID_Q[(\phi)|_\Sigma]$
automatically valid, $IC_Q[\Phi_0] = C_Q[Q[\phi]]|_\Sigma =
C_Q[0]|_\Sigma = 0$, so that, being $Q$-admissible, it coincides with
the unique admissible solution of $Q[\phi] = 0$ determined by $\Phi_0$.
\end{proof}

\subsection{Propagation identity for unit Killing equations} \label{sec:ukilling-propeq}

In this section, we derive a propagation identity for the unit Killing
system~\eqref{eq:ukilling} to which the propagation argument
(Lemma~\ref{lem:propeq}) is applicable. We only discuss that version
where $u_a$ is unit timelike; the spacelike version follows from the
substitution $u_a \mapsto i u_a$ discussed earlier.
Our starting point is the propagation identity form the Killing equation
$K_{ab}[v] = 0$ ($E[\phi] = 0$)~\cite[Sec.2.1]{khgp-ckid}:
\begin{align}
\label{eq:k-propeq}
	\square K_{ab}[v] - 2R^c{}_{ab}{}^d K_{cd}[v]
		&= K_{ab}[Q [v]] + 2 R_{(a}{}^c K_{b)c}[v] - 2 \mathcal{L}_v R_{ab} , \\
		& {}\hspace{5em} (P[E[\phi]] = \sigma[Q[\phi]] + \cdots) \notag
	\\
\label{eq:k-propeq-conv}
	Q_a [v] := \square v_a + R_a{}^b v_b
		&= \nabla^b K_{ab}[v] - \frac{1}{2} \nabla_a K^b{}_b[v] , \\
		& {}\hspace{5em} (Q[\phi] = \rho[E[\phi]]) \notag
\end{align}
where the Lie derivative expands to $\mathcal{L}_v R_{ab} = v^c \nabla_c
R_{ab} + 2 R_{c(a} \nabla_{b)} v^c$, and $\cdots$ refers to terms that
do not fit into either the $P$- or $Q$-expressions unless (cosmological)
vacuum Einstein equations hold. If they do hold, the $\cdots$ terms are
then absorbed into the $P$-expression. Their precise off-shell structure
becomes important when extending the propagation identity to Einstein
equations coupled with matter fields~\cite{racz-kid1, racz-kid2}, which
motivates us to record these terms explicitly, rather than simply
dropping them.

The derivation will proceed as follows. Noting that the equalities
in~\eqref{eq:k-propeq} and~\eqref{eq:k-propeq-conv} hold for arbitrary
$v_a$, we will introduce the parametrization $v_a = \la u_a$ and
systematically cancel the $\la$-dependence from both sides. However, we
must take care not to destroy the desired factorization structure of the
propagation identities, namely that some terms must explicitly factor through $K_{ab}[v]$ and hence through $A_a[\la,u]$~\eqref{eq:A} and
$\bar{K}_{ab}[u]$~\eqref{eq:spK}. Rearranging~\eqref{eq:A} into
$\nabla_a \ln\la = \dot{u}_a - A_a$, after an overall rescaling by
$\lambda$, all terms that dependent on $\la$ and its derivatives will
depend on it through $A_a$. The dependence on $\la$ can then be safely
cancelled without destroying the factorization property by cancelling
only terms that wholly factor through $A_a$. Care must be taken with
higher derivatives. For instance, by~\eqref{eq:B},
\begin{equation}
	\nabla_a A_b = \frac{1}{2} B_{ab} + \frac{1}{2} K_{ab}[A] ,
\end{equation}
where $B_{ab}$ no longer depends on $\la$ and hence need not participate
in the cancellations, while the principal term in $K_{ab}[A]$ is
$\nabla_{(a} \nabla_{b)} \la$ and so cannot be decomposed further. Only
a few second order derivative terms appear in the calculation, and we
obtain the analogous but more complicated identities for the relevant
ones, $\nabla_a \nabla_c A^c$ and $\square A_a$, by contracting the following identity:
\begin{align}
	\notag
	\nabla_a \nabla_b A_c
	&= \nabla_{(a} \nabla_b A_{c)} + \frac{2}{3} \nabla_{(a} B_{b)c}
		+ \frac{2}{3} R_{a(bc)d} A^d
	\\ \notag
	&= \nabla_{(a} \nabla_b A_{c)} + \frac{2}{3} \nabla_{(a} B_{b)c}
	\\ \notag & \quad
	+ \frac{2}{3} (\rho_{a(b:c)d} + \sigma_{a(b:c)} u_d + u_{(b} \sigma_{c)a:d} + u_{(b} \sigma_{c)d:a} + u_a \sigma_{d(b:c)} \notag
	\\ & \qquad
		+ \tau_{a:(b} u_{c)} u_d - \tau_{a:d} u_b u_c - \tau_{b:c} u_a u_d + \tau_{d:(b} u_{c)} u_a) A^d
	, \label{eq:nabla-nabla-A}
\end{align}
where for convenience, we have decomposed the curvature tensor as
\begin{multline}
	R_{abcd}
	= \rho_{ab:cd}
	+ 2(\sigma_{ab:[c} u_{d]} + \sigma_{cd:[a} u_{b]})
	\\
	+ (\tau_{a:c} u_b u_d - \tau_{b:c} u_a u_d
		- \tau_{a:d} u_b u_c + \tau_{b:d} u_a u_c) .
\end{multline}
Here $\rho_{ab:cd}$ has Riemann symmetries (Young type (2,2)), while
$\sigma_{[ab]:c} = \sigma_{ab:c}$ and $\sigma_{[ab:c]} = 0$ (Young type
(2,1)), and $\tau_{(a:b)} = \tau_{a:b}$ (Young type (2)), with each of
these tensors orthogonal to $u_a$ in every index. (The $:$ is a purely
visual separator between convenient index groups.) The individual pieces
can be extracted with the following contractions:
\begin{align*}
	\tau_{a:c} &= R_{abcd} u^b u^d , \\
	\sigma_{ab:c} &= -R_{abcd} u^d - 2\tau_{[a|:c} u_{|b]}
		= -\gamma_a^{a'} \gamma_b^{b'} \gamma_c^{c'} R_{a'b'c'd} u^d , \\
	\rho_{ab:cd} &= R_{abcd}
		- 2(\sigma_{ab:[c} u_{d]} + \sigma_{cd:[a} u_{b]})
	\\ & \quad {}
		- (\tau_{a:c} u_b u_d - \tau_{b:c} u_a u_d
		 - \tau_{a:d} u_b u_c + \tau_{b:d} u_a u_c)
	\\
	&= \gamma_a^{a'} \gamma_b^{b'} \gamma_c^{c'} \gamma_d^{d'} R_{a'b'c'd'} .
\end{align*}
For reference, the $\La$-vacuum Einstein equations decompose as
\begin{equation} \label{eq:lambda-vacuum}
	R_{ab} - \frac{2\La}{n-2} g_{ab}
	= \left(\rho_{ac:b}{}^b - \tau_{a:b} - \frac{2\La}{n-2} \gamma_{ab}\right)
		- 2 u_{(a} \sigma_{b)c:}{}^c
		+ u_a u_{b} \left(\tau_{c:}{}^c + \frac{2\La}{n-2}\right)
	.
\end{equation}

To further manage the complexity of the calculations, we will use
additional constructions. We will need to systematically decompose all
equations into sectors parallel and orthogonal to $u_a$, for which we
introduce the obvious notation $T^*M \cong T_u^*M \oplus T_{\perp}^*M$.
Although $T_{\perp}^*M$ can be identified with a sub-bundle orthogonal
to a unit vector $u_a$, and such an embedding depends on $u_a$, the
images of all such embeddings are isomorphic as smooth bundles,
isomorphic to the abstract bundle $T_{\perp}^*M$. For future
convenience, let us also use the notation $\mathrm{u}T^*M \subset T^*M$
for the non-linear sub-bundle of time-like unit normalized vectors. Next,
we introduce a new \emph{adapted} affine connection $\D_a$, which
preserves not only the metric but also this splitting. It is defined by
\begin{equation} \label{eq:D}
\begin{aligned}
	\D_a S_b = \nabla_a S_b + u^c (\nabla_a u_b) S_c - u_b (\nabla_a u^c) S_c ,
	\\
	\D_a T^b = \nabla_a T^b - u^b (\nabla_a u_c) T^c + u_c (\nabla_a u^b) T^c ,
\end{aligned}
\end{equation}
and the action $\D_a U = \nabla_a U$ on scalars; its extension to higher
rank mixed tensors is uniquely determined by the Leibniz rule. By direct
calculation, it has the promised compatibility properties
\begin{equation}
	\nabla_a (S_b T^b) = (\D_a S_b) T^b + S_b (\D_a T^b) ,
	\quad
	\D_a g_{ab} = 0 ,
	\quad
	\D_a u_b = 0 , \\
	\quad
	\D_a u^b = 0 .
\end{equation}
On the other hand, $\D_a$ is not symmetric (it has torsion). 

We first demonstrate how the cancellation method works on the simplest
case, the factorization of the $Q[\phi]$ operator through the Killing
system~\eqref{eq:k-propeq-conv}. Multiplying through by $\la^{-1}$, the
left-hand side of the equality becomes
\begin{multline} \label{eq:Q-lhs-A}
	\la^{-1} (\square (\la u_a) + \la R_a{}^b u_b)
	= (\D^c + 3\dot{u}^c) \nabla_c u_a
		- (\nabla^c u_c) \dot{u}_a + \sigma_{ab:}{}^b
	\\
		+ u_a [\D^c \dot{u}_c + 2\dot{u}^c \dot{u}_c + (\nabla^c u^b) (\nabla_c u_b)
		- \tau_{c:}{}^c]
	\\
		- \frac{1}{2} u_a K_c{}^c[A]
		- A^c [2\nabla_c + 2\dot{u}_c - A_c] u_a
\end{multline}
The adapted connection $\D_a$ appears after using~\eqref{eq:D} to
re-express each covariant derivative $\nabla_a$ in terms of it (with
$\nabla_a u_b$ itself being the only exception, since $\D_a u_b = 0$).
By the normalization of $u_a$, we have that $\nabla_c u_a$ and hence
$\dot{u}_a$ are orthogonal to $u_a$ in the $a$ index. The use of $\D_a$
then allows to easily keep track of which terms are parallel to $u_a$
and which ones are orthogonal.
The right-hand side of the equality in~\eqref{eq:k-propeq-conv} becomes
\begin{multline} \label{eq:Q-rhs-A}
	\la^{-1} \nabla^c (K_{ca} - \tfrac{1}{2} g_{ca} K_b^b)
	\\
	= (\D^c + \dot{u}^c)
		[\bar{K}_{ca} - \tfrac{1}{2} g_{ca} \bar{K}_b{}^b]
		+ \dot{u}^c \bar{K}_{ca}
		- u^c B_{ca}
		+ u_a (\nabla^b u^c) \bar{K}_{bc}
	\\
		- \frac{1}{2} u_{a} K_c{}^c[A]
		- A^c [2\nabla_c + 2 \dot{u}_c - A_c] u_a
	.
\end{multline}
It is evident that the $\la$-dependence cancels between the left-hand
side~\eqref{eq:Q-lhs-A} and the right-hand side~\eqref{eq:Q-rhs-A}
through identical $A$-dependent terms (the $B$-term does not count,
since it no longer depends on $\la$).

Thus, after cancelling the $\la$-dependence and decomposing it into
$T_{\perp}^*M$ and $T_u^*M$ parts, equality~\eqref{eq:k-propeq-conv}
becomes the following two equalities
\begin{align}
	\notag
	\MoveEqLeft[4]
	\hat{\square} u_a
	:= (\D^c + 3\dot{u}^c) \nabla_c u_a
		- (\nabla^c u_c) \dot{u}_a + \sigma_{ab:}{}^b
	\\ \label{eq:boxhat}
	&= (\D^c + \dot{u}^c)
		[\bar{K}_{ca} - \tfrac{1}{2} \gamma_{ca} \bar{K}_b{}^b]
		+ \dot{u}^c \bar{K}_{ca}
		- u^c B_{ca}
	,
	\\ \notag
	\MoveEqLeft[4]
	\D^c \dot{u}_c + 2\dot{u}^c \dot{u}_c + (\nabla^c u^b) (\nabla_c u_b) 
		- \tau_{c:}{}^c
	\\ \label{eq:boxparl}
	&= \tfrac{1}{2} \dot{\bar{K}}_b{}^b + (\nabla^b u^c) \bar{K}_{bc} 
	.
\end{align}
We have defined the non-linear second order differential operator
$\hat{\square} \colon \Secs(\mathrm{u}T^*M) \to \Secs(T^*_{\perp} M)$ as
the left-hand side of~\eqref{eq:boxhat}. It is determined (the fibers of
$\mathrm{u}T^*M$ and $T_{\perp}^*M$ have the same dimension) and its
principal term (obtained after expanding the adapted derivatives) is
essentially the wave operator $\square$, up to a $u$-dependent embedding
$T_{\perp}^*M \to T^*M$. Later, we shall see that $\hat{\square} u_a$
will play the role of $Q[\phi]$ for our unit Killing system. The
well-posedness properties of the equation $\hat{\square} u_a = 0$ are
postponed until the end of this section.

Now, we expand the identity~\eqref{eq:k-propeq}. Taking advantage
of~\eqref{eq:boxhat} and~\eqref{eq:boxparl} and already hiding all the
$\la$-dependent $A$-terms (knowing that they will cancel), the first term on
the right-hand side becomes (where we denote $(R\cdot u)_a = R_a{}^b
u_b$)%
\begin{align}
	\notag
	\la^{-1} K_{ab}[\square \la u + \la R \cdot u] 
	&= 2(\D + \dot{u})_{(a} \hat{\square} u_{b)} + 2 (u_{(a} \nabla_{b)} u^c) \hat{\square} u_{c}
	\\ \notag & \quad {}
		+ \bar{K}_{ab} [\tfrac{1}{2} \dot{\bar{K}}_c{}^c + (\nabla^c u^d) \bar{K}_{cd}]
	\\ \notag & \quad {}
		+ 2 u_{(a} \nabla_{b)} [\tfrac{1}{2} \dot{\bar{K}}_c{}^c + (\nabla^c u^d) \bar{K}_{cd}]
	\\ \notag & \quad {}
		- 2 B_{(a|}{}^c \nabla_c u_{|b)} - \tfrac{2}{3} u_{(a} (\nabla^c + 3 \dot{u}^c) B_{b)c}
	\\ \label{eq:KB-propeq-Q} & \quad {}
		+ O(A)
	.
\end{align}
The $\cdots$ terms from the right-hand side of~\eqref{eq:k-propeq} are
\begin{multline} \label{eq:KB-propeq-ricci}
	\la^{-1} [2 R_{(a}{}^c K_{b)c}[\la u] - 2 \la u^c \nabla_c R_{ab} - 4 R_{c(a} (\nabla_{b)} \la u^c)]
	\\
	= -2 (\dot{\rho}_{ad:b}{}^d - \dot{\tau}_{a:b}) + 4 u_{(a}\dot{\sigma}_{b)c:}{}^c - 2 u_a u_b \dot{\tau}_{c:}{}^c
	\\
		+ 2(\bar{K}^c{}_{(a} - 2 \nabla_{(a}u^c) (\rho_{b)d:c}{}^d - u_{b)}\sigma_{cd:}{}^d - \tau_{b):c})
		+ O(A)
	.
\end{multline}
Finally, we look at the left-hand side of~\eqref{eq:k-propeq}, the propagation operator for $K_{ab}$, which becomes:
\begin{align}
	\notag
	\MoveEqLeft
	\la^{-1} \square K_{ab}[\la u] - 2 \la^{-1} R^c{}_{ab}{}^d K_{cd}[\la u]
	\\ \notag
	&= (\D^c + \dot{u}^c)(\D_c + \dot{u}_c)\bar{K}_{ab} + (\dot{u}^c\dot{u}_c + \dot{u}^c \D_c)\bar{K}_{ab}
	\\ \notag & \quad {}
		+ 2\bar{B}_{(a|}{}^c \nabla_c u_{|b)} - 2 \dot{u}_{(a} C_{b)} + 2 \bar{K}_{(a|d}(\nabla_c u^d)(\nabla^c u_{|b)})
	\\ \notag & \quad {}
		- u^c [2(\D_d \bar{K}_{(a|c})(\nabla^d u_{|b)}) + (\nabla^d u_d)(\D_c \bar{K}_{ab})]
	\\ \notag & \quad {}
		+2 u_{(a} \big\{ \tfrac{5}{3} \dot{u}^c B_{b)c} + \tfrac{2}{3} \D_c B_{|b)}{}^c
	\\ \notag & \qquad {}
		+ \bar{K}_{|b)c} (\D_d \nabla^d u^c) + C^c \nabla_c u_{|b)} + 3\bar{K}_{|b)d} \dot{u}^c \nabla_c u^d
	\\ \notag & \qquad {}
		- (\nabla^d u_d)[\tfrac{2}{3} B_{|b)c} u^c + \bar{K}_{|b)c} \dot{u}^c] - \tfrac{2}{3} B_{cd} u^c \nabla^d u_{|b)} + 2(\D_d \bar{K}_{|b)c}) \nabla^d u^c
	\\ \notag & \qquad {}
		+\tfrac{1}{2} u_{|b)}[\tfrac{4}{3} B_{cd} \nabla^d u^c + 2\bar{K}_{ce}(\nabla_d u^e)(\nabla^d u^c)] \big\}
	\\ \label{eq:KB-propeq-P} & \quad {}
		+ 2 \bar{K}^{cd} (\rho_{ac:bd} - 2 u_{(b} \sigma_{a)c:d} + u_a u_b \tau_{cd})
		+ O(A)
	,
\end{align}
where we have decomposed $B_{ab} = \bar{B}_{ab} + (u \wedge
C)_{ab} = \bar{B}_{ab} + u_a C_b - u_b C_a$, with $u^a \bar{B}_{ab} =
0$, $u^a C_a = 0$ and $C_a = B_{ab} u^b$.

Taking the expression from~\eqref{eq:KB-propeq-P} and setting it equal
to the sum of the expressions from~\eqref{eq:KB-propeq-Q}
and~\eqref{eq:KB-propeq-ricci}, dropping all the $O(A)$ terms, we obtain
an identity that is almost the desired propagation identity for our unit
Killing system. To get it to the desired form, first all the terms from
the right-hand side that neither factor through $\hat{\square} u_a$, nor
belong to the $\cdots$ terms must be moved to the left-hand side. Then,
both sides of the identity should be decomposed into the part that can
be written as $u_{(a} (-)_{b)}$ and the purely
$u$-orthogonal part, which respectively are
\begin{align}
	\notag
	\MoveEqLeft
	(\nabla^c + 2 \dot{u}^c) B_{bc}
		+ 2 C^c \nabla_c u_{b}
	\\ \notag & \quad {}
		- \nabla_{b} [\tfrac{1}{2} \dot{\bar{K}}_c{}^c + (\nabla^c u^d) \bar{K}_{cd}]
		+ 2(\D_d \bar{K}_{bc}) \nabla^d u^c
	\\ \notag & \quad {}
		+ \bar{K}_{bd} (\D^c + 3 \dot{u}^c) (\nabla_c u^d)
		- \bar{K}^{cd} \sigma_{bc:d}
		- (\nabla^d u_d)\bar{K}_{bc} \dot{u}^c
	\\ \notag & \quad {}
		+ u_{b} \big[\bar{K}_{ce}(\nabla_d u^e)(\nabla^d u^c) + \bar{K}^{cd} \tau_{cd}\big]
	\\
	\label{eq:B-propeq1}
	&= -u^c \D_c \hat{\square} u_{b}
		+ (\nabla_{b} u^c) \hat{\square} u_{c}
	\\ \notag & \qquad {}
		+2 [ u^d \D_d \sigma_{bc:}{}^c - u_b (u^d \D_d \tau_c{}^c + \dot{u}^c \sigma_{cd:}{}^d)  \big]
\end{align}
and
\begin{align}
	\notag
	\MoveEqLeft
	(\D^c + \dot{u}^c)(\D_c + \dot{u}_c)\bar{K}_{ab} + (\dot{u}^c\dot{u}_c + \dot{u}^c \D_c)\bar{K}_{ab}
	\\ \notag & \quad {}
		- u^c [2(\D_d \bar{K}_{(a|c})(\nabla^d u_{|b)}) + (\nabla^d u_d)(\D_c \bar{K}_{ab})]
	\\ \notag & \quad {}
		- \bar{K}_{ab} [\tfrac{1}{2} \dot{\bar{K}}_c^c + (\nabla^c u^d) \bar{K}_{cd}]
		+ 2 \bar{K}_{(a|d}(\nabla_c u^d)(\nabla^c u_{|b)})
		+ 2 \bar{K}^{cd} \rho_{ac:bd}
	\\ \notag & \quad {}
		+ 4 \bar{B}_{(a|}{}^c \nabla_c u_{|b)}
		- 4 \dot{u}_{(a} C_{b)}
	\\
	\label{eq:spK-propeq}
	&= 2 (\D_{(a} + \dot{u}_{(a} + u_{(a} u^c \D_{|c|}) \hat{\square} u_{b)} - 2 u^c \D_c(\rho_{ad:b}{}^d - \tau_{a:b})
	\\ \notag & \qquad {}
		+ 2 (\bar{K}^c{}_{(a} - 2\nabla_{(a} u^c - 2 u_{(a} \dot{u}^c)(\rho_{b)d:c}{}^d - u_{b)} \sigma_{cd:}{}^d - \tau_{b):c})
	.
\end{align}
It is worth remarking that without using the adapted connection $\D_a$,
it would have been less than obvious that the last equation is valued in
${T^*}^{\otimes 2} M$. Expanding the adapted derivatives, the principal
term on the left-hand side of~\eqref{eq:spK-propeq} is just
$\gamma_a^{a'} \gamma_b^{b'}\square \bar{K}_{a'b'}$, which looks
promising for identifying a well-posed evolution operator for
$\bar{K}_{ab}$ and $\bar{B}_{ab}$. Unfortunately, the
equation~\eqref{eq:B-propeq1} is only vector valued (not 2-form valued),
not to mention that its principal term $-\tfrac{1}{2} \nabla_b
\dot{\bar{K}}_c{}^c$ does not even depend on $B_{ab}$, which does not
look helpful for an equation that should evolve $B_{ab}$. However,
taking an exterior derivative of both sides of~\eqref{eq:B-propeq1}, we
find the more promising
\begin{align}
	\notag
	\MoveEqLeft
	-\square_{dR} B_{ab} + 4 (\nabla_{[a} \dot{u}^c) B_{b]c}
		+ 4 \nabla_{[a|} C^c \nabla_c u_{|b]} - 2\nabla_{[a} \bar{K}^{cd} \sigma_{b]:cd}
	\\ \notag & \quad {}
		+ 2 \nabla_{[a|} [(\D^c + 3 \dot{u}^c - (\nabla^e u_e) u^c) (\nabla_c u^d)
			+ 2 (\nabla_c u^d) \D^c] \bar{K}_{|b]d}
	\\ \notag & \quad {}
		+ 2 \nabla_{[a} u_{b]} [(\nabla^d u^c) (\nabla_d u^e) \bar{K}_{ce} + \bar{K}^{cd} \tau_{c:d}]
	\\
	\label{eq:B-propeq}
	&= - 2 \nabla_{[a|} [-u^c \D_c \hat{\square} u
			+ (\nabla u^c) \hat{\square} u_c]_{|b]}
	\\ \notag & \quad {}
		+ 4 \nabla_{[a|} [ u^d \D_d \sigma_{|b]c:}{}^c - u_{|b]} (u^d \D_d \tau_c{}^c + \dot{u}^c \sigma_{cd:}{}^d)  \big]
	,
\end{align}
where
\[
	\square_{dR} B_{ab}
	= (\d \delta B + \delta \d B)_{ab}
	= -2\nabla_{[a} \nabla^c B_{b]c} + 2\nabla^c \nabla_{[c} B_{ab]}
\]
is the Hodge-de~Rham d'Alembertian on $2$-forms, and we have taken
advantage of the fact that definition~\eqref{eq:B} implies that
$\nabla_{[c} B_{ab]} = 0$. The terms on the right-hand side
of~\eqref{eq:B-propeq} that do not factor through $\hat{\square} u_a$
clearly vanish for $\Lambda$-vacua~\eqref{eq:lambda-vacuum}.
Note that the $\nabla_b \dot{\bar{K}}_c^{}c$ from the left-hand side
of~\eqref{eq:B-propeq1} was killed by taking an
exterior derivative. Otherwise it would have been a dangerous term that
added a dependence on $\bar{K}_{ab}$ at 3rd differential order to the
$B$-propagation identity. However, the fact that the $B$-propagation
identity depends on $\bar{K}_{ab}$ at 2nd differential order is still
unusual and it is hard to match a simple notion of hyperbolicity to the
left-hand sides of~\eqref{eq:spK-propeq} and~\eqref{eq:B-propeq}.

Fortunately, with the help from a more exotic kind of hyperbolicity
discussed in Appendix~\ref{app:adn}, we can still conclude
\begin{thm} \label{thm:ukilling-propeq}
When $(M,g_{ab})$ is globally hyperbolic and satisfies the
(cosmological) vacuum Einstein equations, there exists a propagation
identity
\begin{equation}
	\mathcal{P}[u; \bar{K}[u], B[u]] = \sigma[u; \mathcal{Q}[u]] ,
	\quad
	\mathcal{Q}[u] = \rho[u; \bar{K}[u], B[u]]
\end{equation}
for the unit Killing system~\eqref{eq:ukilling}, satisfying the
hypotheses of Lemma~\ref{lem:propeq}, where $\mathcal{Q}[u] =
\hat{\square} u_a$ and $\mathcal{P}$ and $\sigma$ can be read off
from~\eqref{eq:spK-propeq} and~\eqref{eq:B-propeq}, and $\rho$
from~\eqref{eq:boxhat}.
\end{thm}
\begin{proof}
From its definition, it is easy to see that $\rho[u;0] = 0$. When the
Einstein equations are satisfied, a straightforward calculation shows
shows that what we have been calling $\cdots$ terms simplify and factor
through $\bar{K}_{ab}[u]$ and $B_{ab}[u]$, which can then be moved to
the left-hand side in~\eqref{eq:spK-propeq} and~\eqref{eq:B-propeq},
thus determining the $\mathcal{P}$ and $\sigma$ operators. From these
expressions it is obvious that also $\sigma[u;0] = 0$ and $P[u;0,0] =
0$.

When $(M,g_{ab})$ is globally hyperbolic, for any smooth Cauchy surface
$\Sigma$, there exists a smooth Cauchy temporal function $t$ such that
$\Sigma = \{ x \in M \mid t(x) = 0 \}$ (cf.~\cite[Thm.1]{bernard-suhr}
and references therein). It remains to show that $\mathcal{Q}[u]$ and
$\mathcal{P}[u;\bar{K},B]$ have well-posed initial value problems in the
sense of Section~\ref{sec:propeq}, and to identify their admissible
sets. For the former, $\mathcal{A}_\mathcal{Q}$ simply consists of all
unit timelike vector fields $u_a$; $\mathcal{Q}[u] = 0$ is then a
quasi-linear wave equation, whose well-posedness is covered by the
theory in~\cite[Ch.9]{ringstrom-cauchy}. For the latter, we let
$\mathcal{A}_\mathcal{P}$ consist of all pairs $(u, (\bar{K},B))$ with
$u_a$ unit timelike, no condition being imposed on the second slot. The
compatibility hypothesis of Lemma~\ref{lem:propeq} is then immediate,
and in particular $(u,0) \in \mathcal{A}_\mathcal{P}$. This case is more
complicated and is delegated to Lemma~\ref{lem:spKB-propeq-adn} in
Appendix~\ref{app:adn}.
\end{proof}

As discussed in Section~\eqref{sec:propeq}, since
Lemma~\eqref{lem:propeq} applies, we could use~\eqref{eq:spK-propeq}
and~\eqref{eq:B-propeq} to derive purely spatial equations $(\bar{K},
B)^\Sigma[(u)|_\Sigma] = 0$ whose solutions are in bijection with unit
Killing vectors on a neighborhood of $\Sigma$ (imitating for instance
the calculations in~\cite{khgp-ckid}). We have checked that this
procedure reproduces the uKID equations~\eqref{eq:uKID}, but by more
laborious derivation than the one given in Section~\ref{sec:uKID}.%

\section{Discussion} \label{sec:discussion}

In this work, we have introduced the \emph{unit Killing equations}, a non-linear system
characterizing unit vector fields proportional to Killing vectors, and
derived the corresponding \emph{unit Killing initial data} (uKID)
equations. The latter provide necessary and sufficient conditions for a
$\Lambda$-vacuum Einstein initial data set to admit, in any globally
hyperbolic development, a unit vector field locally proportional to a
time-like or space-like Killing vector. In contrast to the classical KID
equations, the uKID system is intrinsically non-linear because of the
unit normalization constraint. Nevertheless, we have shown that it
remains a finite-type system by constructing an explicit prolongation to
canonical connection form, which immediately bounds the dimension of its
solution space by $n(n+1)/2-1$.

Besides the direct derivation of the uKID equations, we have also
revisited the propagation identity method used previously for Killing
and conformal Killing initial data. The unit Killing system requires a
substantial extension of that framework. In particular, the natural
propagation operator couples the spatial symmetric tensor
$\bar{K}_{ab}$ and the $2$-form $B_{ab}$ through second-order terms,
preventing the resulting system from fitting the standard theory of
normally hyperbolic equations. To overcome this difficulty, we developed
a more general formulation of the propagation argument and showed that
the relevant evolution system is well posed by appealing to
Agmon-Douglis-Nirenberg hyperbolicity. Besides establishing an independent derivation of the
uKID equations, this broader framework should also be applicable to
other geometric PDEs whose natural propagation systems fail to be
normally hyperbolic.

The principal motivation for introducing the unit Killing equations is
that, in many physically relevant situations, the geometrically natural
object is not the Killing vector itself but only its normalized
direction. This occurs, for example, for static spherically symmetric perfect-fluid solutions,
where the unit direction of the time-like Killing vector field can be reconstructed intrinsically from curvature invariants, while the norm of the Killing vector is not. Consequently, the uKID equations
provide the missing ingredient needed to formulate intrinsic initial
data characterizations for such spacetimes, in the same spirit that the
classical KID equations enter the geometric characterization of the Kerr
solution.

Several natural extensions remain open. The most immediate one is to
replace the vacuum Einstein equations by Einstein equations coupled to
matter. Since the off-shell structure of the propagation identities was
kept throughout the derivation, adapting the argument to specific matter
models, such as perfect fluids, should be feasible. This would be a
necessary step toward intrinsic initial data characterizations of
physically important fluid solutions, including static stellar
models.

More broadly, the generalized propagation framework developed here may
prove useful for other overdetermined geometric systems whose
propagation equations possess a more complicated principal structure.
Whether similar techniques can be applied to other non-linear variants
of the Killing equations, or to higher-rank geometric PDEs admitting
finite-type prolongations, remains an interesting direction for future
research.

\section*{Acknowledgements}
IK's was partially supported by GA\v{C}R project GA22-00091S and the \emph{Czech Academy of Sciences} Research Plan RVO: 67985840. SM was partially supported by the Generalitat Valenciana grant CIACIF/2021/028, the Generalitat Valenciana Project AICO/2020/125 and the Universitat Polit\`ecnica de Val\`encia Project PAID-11-25. SM would also like to thank the Institute of Mathematics of the Czech Academy of Sciences for its hospitality during the two research stays in which this work was partially carried out.

\appendix

\section{ADN hyperbolic PDE systems} \label{app:adn}

The evolution equations that appear in the propagation identity for the unit
Killing equation in Section~\ref{sec:ukilling-propeq} are of \emph{mixed order}
and not of an obvious classical hyperbolic type. However, they do fall into the
class of \emph{Agmon-Douglis-Nirenberg (ADN) strictly hyperbolic} systems. Our
main reference about mixed-order hyperbolic systems is the
monograph~\cite{gv-mixed}, where ADN hyperbolicity is discussed. The main
relevant definitions and properties are summarized below.

Recall that the \emph{principal symbol} of a linear differential
operator $L[w]$ on $\mathbb{R}^n$ of order $k$ consists of all the terms
containing derivatives of order $k$ and, replacing each occurrence of a
partial derivative $\del_i$ by the indeterminate $p_i$, is denoted by
$L^0(p)$.

Note that we have tied the notion of the principal symbol to
the order $k$, so that when an operator $L[w]$ of order $k$ is
considered as a degenerate case of an operator of a larger order $k'>k$, its
principal symbol as an order $k'$ operator vanishes. When $L[w]$ is of
order $k$, it cannot be considered a differential operator of lower order $k'$
for any $k'<k$, and hence its principal symbol as an order $k'$ operator
is undefined.

\begin{dfn}[{\cite[\textsection\textsection I.2.2, II.4.1, II.7.2.1]{gv-mixed}}]
\label{dfn:strictly-hyper}
Let $M = \mathbb{R}^n$ and fix one of the standard coordinates to be the
time variable $t$. Denote the corresponding derivative by $\del_t$,
and the derivative in the complementary standard coordinates by $\del_\perp$.

(a) A linear differential operator $y = L[w]$ is called
\emph{(uniformly) strictly hyperbolic} if the polynomial
$\det L^0(p)$ is (i) monic (has unit coefficient) with respect to the
highest power of $p_t$,%
	\footnote{The original definition~\cite[\textsection I.2.2]{gv-mixed} stipulates that the coefficient of the highest power of $p_t$ is exactly $1$, but from context it is clear that we can relax this condition for the coefficient to be an invertible element (a \emph{unit}) in the ring $C^\oo(\mathbb{R}^n)[p_t,p_\perp]$, that is, any smooth position dependent scalar that does not change sign.} %
(ii) all of its roots with respect to $p_t$
(all other variables are fixed) are real and simple, and (iii) these
roots are separated from each other uniformly by at least $\delta
\|p_\perp\|$, for some $\delta > 0$, over all spatial momenta
$p_\perp$ and spacetime points.

(b) Consider a linear differential operator $y = L[w]$, with the following
block structure under the decomposition $y = (y_1,\dots,y_m)$, $w =
(w_1,\dots,w_m)$:
\begin{equation}
	L[w] = \begin{bmatrix}
			L_{11} & \cdots & L_{1m} \\
			\vdots & & \vdots \\
			L_{m1} & \cdots & L_{mm}
		\end{bmatrix}
		\begin{bmatrix}
			w_1 \\ \vdots \\ w_m
		\end{bmatrix}
	.
\end{equation}
We say that $L$ is \emph{Agmon-Douglis-Nirenberg (ADN) strictly
hyperbolic} when there exist integer \emph{weights} $t_j, s_j \ge 0$, $j=1,\ldots,m$,
such that $L_{ij}$ is a differential operator of order $(t_j-s_i)$ and
the \emph{weighted principal symbol} has the block diagonal structure
\begin{equation}
	L^0(p) := \begin{bmatrix}
			L_{11}^0(p) & \cdots & L_{1m}^0(p) \\
			\vdots & & \vdots \\
			L_{m1}^0(p) & \cdots & L_{mm}^0(p)
		\end{bmatrix}
		= \begin{bmatrix}
			L_{11}^0(p) & & \\
			& \ddots & \\
			& & L_{mm}^0(p)
		\end{bmatrix}
	,
\end{equation}
with $L_{ij}^0(p)$ being the principal symbol of $L_{ij}$ as an operator
of order $(t_j-s_i)$, where each $L_{ii}^0(p)$ is strictly hyperbolic.
\end{dfn}
Note that according to the above definition any scalar operator that has the
wave operator $\square = g^{ab} \nabla_a \nabla_b$ as its principal term is
strictly hyperbolic with respect to a Cauchy time slicing on a globally
hyperbolic background. The chiral Dirac equation $\not{\nabla}_\pm =
\gamma_\pm^a \nabla_a$ and some first order symmetric hyperbolic
systems~\cite[Def.10]{kh-chargeom} (those with simple characteristics) are also
strictly hyperbolic, justifying that not only scalar operators are examples.
By extension any system whose principal term is, for instance, $\square$
multiplied by the identity matrix is then ADN strictly hyperbolic.

The following well-posedness result can be stated for the initial value
problem of ADN strictly hyperbolic systems.
\begin{dfn}[{\cite[\textsection II.7.2.4]{gv-mixed}}] \label{dfn:mixed-ivp}
Assuming the notation and hypotheses of
Definition~\ref{dfn:strictly-hyper}, for the inhomogeneous PDE system
$L[w] = y$, the \emph{mixed order initial data} consists of smooth
functions $(w_j^{(l)})_{j=1,\ldots,m}^{l=0,\ldots,t_j-1}$ on the
surface $t=0$. The corresponding \emph{mixed-order initial value
problem} consists of finding a smooth $w$ that satisfies $L[w]=y$ for
all time and the conditions
\begin{equation} \label{eq:mixed-ivp}
\begin{aligned}
	\del_t^l w_j|_{t=0} &= w_j^{(l)}
		~ (l=0,\ldots,t_j-1)
	\\
	\text{and} \quad
	\del_t^l \sum_{j=1}^m L_{ij}[w_j]|_{t=0} &= \del_t^l y_i|_{t=0}
		~ (l=0,\ldots,s_i-1) ,
\end{aligned}
\end{equation}
where the second equality condition is empty whenever $s_i=0$, and
otherwise its left-hand side should be thought of as an expression
purely in terms of the $(w_j^{(l)})$ data obtained by formal time
differentiation. The above initial value problem is \emph{well-posed}
when there exists for all time a unique smooth $w$ satisfying $L[w]=y$
and the conditions~\eqref{eq:mixed-ivp}.
\end{dfn}

Note that, when $y=0$, due to the linearity of the equation $L[w] = 0$, if the
initial data $(w_j^{(l)})_{j=1,\ldots,m}^{l=0,\ldots,t_j-1} = 0$, the
unique solution is $w=0$.

\begin{prp}[{\cite[\textsection II.7.2.4]{gv-mixed}}] \label{prp:adn-well-posed}
Let $M = \mathbb{R}^n$ and fix one of the standard coordinates as the
time variable $t$, such that $t=0$ is a Cauchy surface. Consider a
uniformly ADN strictly hyperbolic system $L[w] = y$ with smooth
coefficients (Definition~\ref{dfn:strictly-hyper}). Then its initial
value problem (as in Definition~\ref{dfn:mixed-ivp}) is well-posed.
\end{prp}

Let us spell out how Definition~\ref{dfn:mixed-ivp} and
Proposition~\ref{prp:adn-well-posed} supply, for the homogeneous system $L[w] =
0$, all the ingredients of a well-posed initial value problem in the
terminology of Section~\ref{sec:propeq}. The initial data operator is $\Psi_0 =
ID[w] = (\del_t^l w_j)_{j=1,\ldots,m}^{l=0,\ldots,t_j-1}$, while the operator
$C[y] = (\del_t^l y_i)_{i=1,\ldots,m}^{l=0,\ldots,s_i-1}$ (which obviously
satisfies $C[0] = 0$) captures the second condition in~\eqref{eq:mixed-ivp}.
The initial constraint operator $IC[\Psi_0]$ is obtained by rewriting the
left-hand side of that condition purely in terms of the initial data $\Psi_0 =
(w_j^{(l)})$ by formal time differentiation, so that the compatibility relation
$IC[ID[(w)|_{t=0}]] = C[L[w]]|_{t=0}$ holds identically. Note that the number
of formal time derivatives of each component $w_j$ included in $ID[w]$ is
precisely the smallest for which the compatibility relation can be expressed
without any further time differentiation. Together with the observation made
above that vanishing initial data give rise to the vanishing solution, this
makes the homogeneous system $L[w] = 0$ a \emph{propagation equation} in the
sense of Section~\ref{sec:propeq}.

We have not found a version of Definitions~\ref{dfn:strictly-hyper},
\ref{dfn:mixed-ivp} and Proposition~\ref{prp:adn-well-posed} stated directly
for a general manifold $M$. It would go beyond the scope of this work to work
them out in detail. We simply formulate in
Proposition~\ref{prp:adn-global-well-posed} a reasonable generalization based
on geometric structures that arise in studies of some other general classes of
hyperbolic systems~\cite{kh-chargeom}. Accordingly, we do not give it a
complete proof, but merely sketch how it should proceed following some standard arguments.

\paragraph{Principal symbol.}
The first thing to note is that the principal symbol $L^0(p)$ of a linear
differential operator $L[w]$ is well-defined on manifolds. When $\mathbb{R}^n$
is replaced by a manifold $M$, a matrix valued differential operator is
replaced by a smooth linear differential operator $L \colon \Secs(W) \to
\Secs(V)$ that maps between sections of vector bundles $V\to M$ and $W\to M$.
The indeterminates $p_i$ are interpreted as coordinates on the fibers of the
cotangent bundle $T^*M$, and the replacement of the partial derivatives
$\del_i$ by $p_i$ can be done in any coordinate chart. The transformation rules
for differential operators then guarantee that the principal symbol becomes a
well-defined fiber-wise homogeneous polynomial function on $T^*M$ valued in the
linear vector bundle morphisms $L^0 \colon T^*M \to \operatorname{Hom}(W,V)$.

The principal symbol of a scalar wave operator $L[w] = \square w$ is of
course $L^0(p) = p_a p_b g^{ab}$.

\paragraph{Characteristic covector cone, time orientation.}
Next, a core observation is that the locus of solutions $p = (p_t, p_\perp) \in
T^*_xM$ of the polynomial equation $\det L^0(p) = 0$ is well-defined also on
manifolds and for a strictly hyperbolic system $L[w]=0$ on $\mathbb{R}^n$ it is
a possibly multi-sheeted (\emph{characteristic}) surface
$\mathcal{C}_x^\oast(L^0) \subset T^*_xM$ that consists of properly nested
cones, say $\mathcal{C}_x^\oast(L^0) = \bigcup_{j=0}^r
\mathcal{C}_x^{j,\oast}(L^0)$. The root separation condition implies that each
cone sheet is a simple root of $\det L^0(p) = 0$ and is smooth away from the
tip $p=0$. Because $\det L^0(p)$ is a homogeneous polynomial,
$\mathcal{C}_x^\oast(L^0)$ is symmetric under $p\mapsto -p$ reflection, and
$\mathcal{C}_x^{j,\oast}(L^0) \setminus \{0\} = \mathcal{C}_x^{+j\oast}(L^0)
\sqcup \mathcal{C}_x^{-j,\oast}(L^0)$ splits into two halves interchanged by
the reflection. The \emph{future oriented} $(+)$ and \emph{past oriented} $(-)$
labels are a priori arbitrarily assigned. We let $j=0$ label the inner-most
cone sheet (the one with the largest value of $|p_t|$ for given $p_\perp$).
This structure of $\mathcal{C}_x^\oast(L^0)$ is almost enough to reasonably
define strictly hyperbolic operators on manifolds, with the caveat that it does
not cover the monic structure of $\det L^0(p)$ and the $x$-uniform separation
between the cone sheets $\mathcal{C}_x^{j,\oast}(L^0)$, which are not easy to
define invariantly on a general manifold. The union $\mathcal{C}^{\oast}(L^0) =
\bigsqcup_{x\in M} \mathcal{C}_x^{\oast}(L^0) \subset T^*M$ is the globally
defined \emph{characteristic covector} variety. The unions of the individual
sheets are defined similarly and give the nested cone varieties
$\mathcal{C}^{j,\oast}(L^0) \subseteq \mathcal{C}^\oast(L^0)$. In order to
avoid an overly complicated geometric formalism, we simply require our notion
of a strictly hyperbolic system $L[w] = y$ on a manifold $M$ to have a variety
$\mathcal{C}^\oast(L^0) \subset T^*M$ that is a smooth submanifold (excluding the
$p=0$ section) that consist of pointwise properly nested cones, such that there
also exists a locally finite cover of $M$ by charts, with trivializations of
$V$ and $W$, such that in each chart the monic and root separation conditions
on $\det L^0(p)$ are verified uniformly. We will say that two systems $L_1[w_1]
= y_1$ and $L_2[w_2] = y_2$ are \emph{simultaneously} strictly hyperbolic when
they verify the above conditions with respect to each chart on a common locally
finite cover, which easily generalizes to any finite number of systems. The
characteristic covector variety is \emph{time orientable} when a coherent
choice of $(+)$ and $(-)$ halves of the inner characteristic covector cones can
be made globally on $M$, that is more precisely when
$\mathcal{C}^{0,\oast}(L^0) \setminus \{0\} = \mathcal{C}^{+0,\oast}(L^0)
\sqcup \mathcal{C}^{-0,\oast}(L^0)$ decomposes into a disjoint union of
connected sub-varieties (assuming $M$ is itself connected, and working on each
connected component otherwise). The choice of such decomposition constitutes a
\emph{time orientation}.

Again, for the scalar wave operator $\det L^0(p) = p_a p_b g^{ab}$ and
$\mathcal{C}^\oast(L^0) \subset T^*M$ is the standard covector null-cone.

\paragraph{ADN weights.}
The block decompositions of $L$ and $L^0(p)$ on a manifold must be interpreted
with respect to some direct decompositions $V = \bigoplus_i V_i$, $W =
\bigoplus_j W_j$ of the vector bundles. The ADN weights $s_j$ and $t_j$ are
assigned respectively to the sub-bundles $V_j$ and $W_j$. On a general manifold
$M$, we can take the definition of an ADN strictly hyperbolic system $L[w] = 0$
to consist of (a) a block decomposition, (b) the weights assigned to the
blocks, (c) the bound on the differential order of each $L_{ij}$ block, and (d)
the condition that the weighted principal symbol $L^0_{ij}(p)$ is block
diagonal with the diagonal symbols $L^0_{jj}(p)$ being simultaneously strictly
hyperbolic on $M$. An ADN strictly hyperbolic system is time orientable when
each of the strictly hyperbolic diagonal principal symbols is time orientable
and a time orientation consists of a simultaneous choice of time orientation
for each diagonal block.

\paragraph{Ray vector cone, global hyperbolicity.}
Each $\mathcal{C}^{j,\oast}(L^0)$ sheet has a \emph{polar} (or
\emph{projective}) dual (\emph{ray}) cone $\mathcal{C}_x^j(L^0) \subset
T_xM$~\cite{wiki-projective-dual}, consisting of those vectors $\xi \in T_xM$
that represent tangent hyperplanes to $\mathcal{C}_x^{j,\oast}(L^0)$. The unions
$\mathcal{C}^{j}(L^0) = \bigsqcup_{x\in M} \mathcal{C}_x^{j}(L^0)$ are defined
similarly and $\mathcal{C}^{0}(L^0) \subset TM$, where $j=0$ now labels the
outer-most cone sheet, is the \emph{outer ray vector}
variety~\cite[Def.16]{kh-chargeom}. A time orientation also implies a
corresponding decomposition $\mathcal{C}^0(L^0) = \mathcal{C}^{+0}(L^0) \sqcup
\mathcal{C}^{-0}(L^0)$. For an ADN strictly hyperbolic system, we take
$\mathcal{C}^{\pm0}(L^0) = \bigcup_j \mathcal{C}^{\pm0}(L^0_{jj})$. The
fiberwise closed convex hull $\bar{\mathcal{C}}^{\pm0}(L^0)$ of
$\mathcal{C}^{\pm0}(L^0)$ defines
for instance the allowed
tangents of a future ($+$) or past ($-$) directed causal curves. The same cones
also set the maximum speed of propagation for solutions of $L[w]= y$. The
\emph{cone structure} $\bar{\mathcal{C}}^{+0}(L^0) \subset TM$ defines a
\emph{causal structure} on $M$~\cite{minguzzi}. We will not go into any details
of that formalism here, we shall only borrow from it the notion of \emph{global
hyperbolicity}~\cite[Def.2.20]{minguzzi} in the guise of the existence of a
\emph{Cauchy temporal function} $t$~\cite[Thm.2.45]{minguzzi}, which is a
smooth function on $M$ such that $dt(\xi) > 0$ for any non-zero $\xi \in
\bar{\mathcal{C}}_x^{+0}(L^0)$ and each inextensible causal curve of
appropriate regularity passes through each level set of $t$ exactly once.

At this level of generality, following~\cite{gv-mixed}, we have not asked for
any convexity of the $\mathcal{C}_x^{\pm j,\oast}$ cones. However, when they
bound convex cones, the duality relation to the corresponding ray cone is the
perhaps more familiar \emph{polar} (or \emph{convex})
duality~\cite[Def.15]{kh-chargeom}, namely $\mathcal{C}_x^{\pm j}(L^0)$ is the
boundary of the set of vectors $\xi \in T_xM$ whose corresponding half-spaces
in $T_x^*M$ contain $\mathcal{C}_x^{j,\oast}(L^0)$. For the usual covector
null-cone $p_a p_b g^{ab} = 0$ its dual is just the usual vector null-cone
$p^a p^b g_{ab} = 0$.

\paragraph{Quasi-linear systems.}
Next, we need to extend the above notions to some reasonable non-linear PDE
systems. Namely, we call a PDE system $L[w] = 0$ \emph{quasi-linear ADN
strictly hyperbolic} when (a) there exists a differential operator $N[w]$ and a
parametrized differential operator $\ell[w;v]$, linear in $v$, such that $L[w]
= \ell[w;w] + N[w]$, (b) there exists a block decomposition $L_i[w] = \sum_j
\ell_{ij}[w;w_j] + N_i[w]$ with corresponding ADN weights $s_j$, $t_j$ and an
\emph{$L$-admissible} subset $\mathcal{A} \subset \Secs(W)$ such that
$\ell[w;v] = \bigoplus_i \sum_j \ell_{ij}[w;v_j]$ is ADN strictly hyperbolic
when $w\in \mathcal{A}$, and (c) the differential order of the dependence on
the component $w_j$ of the coefficients of the $i$-th block row $\sum_k
\ell_{ik}[w;v_k] + N_i[w]$ is strictly $< (t_j-s_i)$.

The role of the $L$-admissible subset $\mathcal{A}$ is analogous to that of the
$P$-admissible parameter fields from Section~\ref{sec:propeq}: it delimits the
configurations $w$ about which the quasi-linear structure of $L[w]$ remains ADN
strictly hyperbolic. In the classical quasi-linear wave equation $g^{ab}
\partial_a \partial_b g_{cd} + N(g,\partial g) = 0$ obtained by a harmonic
coordinate gauge fixing of the Einstein
equations~\cite[Sec.14.1]{ringstrom-cauchy}, the admissible metric
configurations $g_{ab}$ are precisely those that have Lorentzian signature.

\begin{prp} \label{prp:adn-global-well-posed}
Let $L\colon \Secs(W) \to \Secs(V)$ be a quasi-linear ADN strictly hyperbolic PDE
system on $M$, for which there exists a Cauchy temporal function $t$. Then,
with the initial data restricted to those induced by $L$-admissible
configurations $w \in \mathcal{A}$, the system $L[w] = 0$ has a well-posed
initial value problem in the sense of Section~\ref{sec:propeq}, with the
solution in general guaranteed to exist only on some neighborhood $U \supset
\Sigma$ of the initial Cauchy surface $\Sigma = t^{-1}(0)$.
\end{prp}

\begin{proof}[Sketch of proof]
For each chart of the locally finite cover (appearing in the definition of
simultaneous strict hyperbolicity), the linear system $\ell[w_{(0)};w_{(1)}]
= y$ with a frozen admissible configuration $w_{(0)} \in \mathcal{A}$ is an
ADN strictly hyperbolic system. The initial data and the inhomogeneity can be
broken up by a partition of unity to have compact support in each chart. The
PDE system in each chart admits a unique solution for any such valid initial
data on the Cauchy surface $\Sigma = t^{-1}(0)$, by
Proposition~\ref{prp:adn-well-posed}. Since hyperbolic systems have finite
speed of propagation, bounded by the cone structures
$\bar{\mathcal{C}}^\pm(L^0)$, reconstructing the solution by summing over the
charts, at any point close to $\Sigma$ only finitely many charts will
contribute, giving a well-defined result. Shifting the initial data surface
$\Sigma$ to the future in this domain allows the argument to be repeated,
expanding the domain of the solution as much as needed.

A local-in-time solution of the quasi-linear problem $L[w] = 0$ is constructed
by standard contraction mapping or iterative arguments (see for
instance~\cite[Ch.9]{ringstrom-cauchy}, where second order quasi-linear wave
equations are treated). The iteration proceeds by solving the linearized
system $\ell[w_{(k)};w_{(k+1)}] = 0$ with coefficients frozen at the $k$-th
iterate $w_{(k)}$, and the monotonicity and uniform separation of the cone
sheets guarantee convergence. Each iterate remains in $\mathcal{A}$ provided
the initial data are sufficiently small or the time interval is sufficiently
short.
\end{proof}

Finally, let us sketch how to check that the evolution equations
obtained at the end of Section~\ref{sec:ukilling-propeq} are indeed of
the quasi-linear ADN strictly hyperbolic type discussed above.
\begin{lem} \label{lem:spKB-propeq-adn}
Let $(M,g_{ab})$ be globally hyperbolic and let $u_a$ be a unit timelike
vector field on $M$. Then the propagation operator
$\mathcal{P}[u;\bar{K},B]$, whose two block components $\mathcal{P}_1$
and $\mathcal{P}_2$ are the left-hand sides of
respectively~\eqref{eq:spK-propeq} and~\eqref{eq:B-propeq}, is
quasi-linear ADN strictly hyperbolic in the unknowns $\psi =
(\bar{K}_{ab}, B_{ab})$ with respect to the weights
\begin{equation} \label{eq:spKB-weights}
	(t_1,t_2) = (3,2) \quad \text{for } (\bar{K}_{ab}, B_{ab}) ,
	\qquad
	(s_1,s_2) = (1,0) \quad \text{for } (\mathcal{P}_1, \mathcal{P}_2) ,
\end{equation}
and with admissible set $\mathcal{A}_\mathcal{P}$ consisting of all
pairs $(u,\psi)$ with $u_a$ unit timelike and $\psi = (\bar{K}_{ab},
B_{ab})$ arbitrary. Moreover, $\mathcal{P}[u;\bar{K},B] = 0$ is a
propagation equation in the sense of Section~\ref{sec:propeq}; its
initial data operator is
\begin{equation} \label{eq:spKB-ID}
	ID_{\mathcal{P}}[\bar{K},B]
	= (\bar{K}, \del_t \bar{K}, \del_t^2 \bar{K}; B, \del_t B)|_\Sigma ,
\end{equation}
subject to the single initial constraint is $\mathcal{P}_1[u;\bar{K},
B]|_{\Sigma} = 0$.
\end{lem}

\begin{proof}
The unknowns are sections of $W_1 = S^2 T_\perp^* M$ (for the
$u$-orthogonal symmetric $\bar{K}_{ab}$) and $W_2 = \wedge^2 T^* M$
(for $B_{ab}$), and the two equations are valued in the same pair of
bundles, $V_i = W_i$. Since $\D_a$ annihilates both $g_{ab}$ and
$u_a$, it preserves the splitting $T^*M \cong T_u^*M \oplus
T_\perp^*M$, so the diagonal blocks below really do act within $W_1$
and $W_2$. In the principal symbol, $\D_a$ can be directly replaced by
$\del_a$. Any terms that could be moved to the left-hand side from the
$\cdots$ terms in~\eqref{eq:spK-propeq} and~\eqref{eq:B-propeq} can be
ignored for our purposes, since they would not affect the principal
symbol.

\emph{Block orders.} The weights~\eqref{eq:spKB-weights} require the
blocks $L_{ij}$ to have orders
\begin{equation}
	\begin{bmatrix} t_1-s_1 & t_2-s_1 \\ t_1-s_2 & t_2-s_2 \end{bmatrix}
	= \begin{bmatrix} 2 & 1 \\ 3 & 2 \end{bmatrix} ,
\end{equation}
and inspection of~\eqref{eq:spK-propeq} and~\eqref{eq:B-propeq} shows
that each block stays within its bound: $\mathcal{P}_1$ involves
$\bar{K}_{ab}$ at order $2$ and $B_{ab}$ only at order $0 < 1$, while
$\mathcal{P}_2$ involves $\bar{K}_{ab}$ at order $2 < 3$ and
$B_{ab}$ at order $2$. Both off-diagonal blocks thus fall strictly below
their bound, their weighted principal symbols vanish, and $L^0(p)$ is
block diagonal.%
	\footnote{This is precisely where taking the exterior derivative
	of~\eqref{eq:B-propeq1} pays off: the term $\nabla_b
	\dot{\bar{K}}_c{}^c$ that it eliminated would have been of order $3 =
	t_1 - s_2$ and would have contributed a non-vanishing off-diagonal
	symbol.}

\emph{Diagonal blocks.} The principal part of the $(1,1)$ block is
$\D^c\D_c$ and that of the $(2,2)$ block is $-\square$. Hence
\begin{equation}
	\mathcal{P}^0(p) = (p_a p_b g^{ab})
		\begin{bmatrix} \id_{W_1} & 0 \\ 0 & -\id_{W_2} \end{bmatrix} ,
\end{equation}
that is, both diagonal symbols are that of the scalar wave operator
$\square$ times identity (up to a harmless sign). We apply
Definition~\ref{dfn:strictly-hyper}(b) with the blocks refined down to
individual scalar components in any local trivialization, each
$1\times 1$ block then being a scalar wave operator $\square$, which
is strictly hyperbolic with respect to a Cauchy time slicing by the
remark following Definition~\ref{dfn:strictly-hyper}. Individual
scalar components of a given tensor field inherit the same weight in
this refinement. Since both diagonal blocks share the single-sheeted
characteristic covector variety $\mathcal{C}^\oast(\mathcal{P}^0) = \{ p_a p_b
g^{ab} = 0 \}$, simultaneous strict hyperbolicity holds with respect
to any locally finite cover by charts adapted to a Cauchy time
slicing; time orientability reduces to that of $(M,g_{ab})$; the cone
structure $\bar{\mathcal{C}}^{+0}(\mathcal{P}^0)$ is the closed solid future
light cone of $g_{ab}$, and a Cauchy temporal function exists by
global hyperbolicity of $(M,g_{ab})$~\cite[Thm.2.45]{minguzzi}.

\emph{Quasi-linearity.} Quasi-linearity holds by direct inspection
of~\eqref{eq:spK-propeq} and~\eqref{eq:B-propeq}. The symbol
computation above is independent of $\psi$ and $u_a$, so $(u, \psi)
\in \mathcal{A}_\mathcal{P}$ for all pairs of field configurations.

\noindent We have now verified all the conditions for $\mathcal{P}[u;
\bar{K}, B]$ to be a quasi-linear ADN strictly globally hyperbolic
equation. Lemma~\ref{prp:adn-global-well-posed} essentially states that
it is a propagation equation in the sense of Section~\ref{sec:propeq},
provided we complete it with the following:

\emph{Initial data.} Definition~\ref{dfn:mixed-ivp} with the
weights~\eqref{eq:spKB-weights} prescribes $t_1 = 3$ time derivatives of
$\bar{K}_{ab}$ and $t_2 = 2$ of $B_{ab}$ on $\Sigma$, which is
$ID_\mathcal{P}$ in~\eqref{eq:spKB-ID}, together with the single $s_1 =
1$ constraint $\mathcal{P}_1[u;\bar{K},B]|_\Sigma = 0$ from the first
block and none ($s_2 = 0$) from the second.

\emph{Compatibility.} It is obvious from its definition that
$\mathcal{P}[u;0, 0] = 0$.
\end{proof}

\begin{rem}
Much of the geometric structure introduced earlier in this appendix
degenerates in this example: all the diagonal blocks share the same
single-sheeted characteristic covector variety, the null cone of
$g^{ab}$, so that the causal cone structure
$\bar{\mathcal{C}}^{+0}(\mathcal{P}^0)$ and its Cauchy temporal
functions are just the usual ones of $(M,g_{ab})$. The extra
generality of Appendix~\ref{app:adn} is spent entirely on the mixed
orders, i.e.,\ on the weights. However, this level of generality is
difficult to avoid if one wants to refer to the theory as it is
presented in~\cite{gv-mixed}.
\end{rem}

\bibliographystyle{utphys-alpha}
\bibliography{kid}

\end{document}